\documentclass[journal]{IEEEtran}

\IEEEoverridecommandlockouts

\usepackage{cite}
\usepackage{graphicx}
\usepackage{picinpar}
\usepackage{stfloats}
\usepackage{url}
\usepackage[latin1]{inputenc}
\usepackage{colortbl}
\usepackage{soul}
\usepackage{multirow}
\usepackage{pifont}
\usepackage{alltt}
\usepackage{enumerate}
\usepackage{siunitx}
\usepackage{pbox}
\usepackage{amsmath,amssymb,amsfonts}
\usepackage{amsthm}
\usepackage{algorithmic}
\usepackage{algorithm}
\usepackage{graphbox}
\usepackage{textcomp}
\usepackage{xcolor}
\usepackage[export]{adjustbox}
\usepackage{color}
\usepackage{multicol}
\usepackage{tikz}
\def\BibTeX{{\rm B\kern-.05em{\sc i\kern-.025em b}\kern-.08em
    T\kern-.1667em\lower.7ex\hbox{E}\kern-.125emX}}

\definecolor{LightCyan}{rgb}{1,0.8,0.8}

\newtheorem{theorem}{Theorem}
\newtheorem{corollary}{Corollary}

\newtheorem{proposition}{Proposition}
\newtheorem{assumption}{Assumption}

\makeatletter
\let\NAT@parse\undefined
\makeatother
\usepackage[hidelinks]{hyperref}

\newcommand\copyrighttext{%
  \footnotesize \textcopyright This paper has been accepted for publication in the IEEE Transactions on Robotics. Please cite the paper as: E. Sebasti\'{a}n, E. Montijano, and S. Sag\"{u}\'{e}s,``Adaptive Multi-robot Implicit Control of Heterogeneous Herds'', IEEE Transactions on Robotics (T-RO), 2022.}
\newcommand\copyrightnotice{%
\begin{tikzpicture}[remember picture,overlay]
\node[anchor=south,yshift=10pt] at (current page.south) {\fbox{\parbox{\dimexpr\textwidth-\fboxsep-\fboxrule\relax}{\copyrighttext}}};
\end{tikzpicture}%
}

\title{Adaptive Multi-robot Implicit Control of Heterogeneous Herds}

\author{Eduardo Sebasti\'{a}n, Eduardo Montijano and~Carlos Sag\"{u}\'{e}s\thanks{E. Sebasti\'{a}n, E. Montijano and C. Sag\"{u}\'{e}s are with the RoPeRT group, at DIIS - I3A, Universidad de Zaragoza, Spain.
\texttt{\small \{esebastian, emonti, csagues\}@unizar.es}} 
\thanks{This work has been supported by the ONR Global
grant N62909-19-1-2027, the Spanish projects PGC2018-098817-A-I00 and PGC2018-098719-B-I00 (MCIU/AEI/FEDER, UE), DGA T45-20R, and Spanish grant FPU19-05700.}
}

\begin{document}
\maketitle

\copyrightnotice

\begin{abstract}
This paper presents a novel control strategy to herd groups of non-cooperative evaders by means of a team of robotic herders. In herding problems, the motion of the evaders is typically determined by \textit{strongly nonlinear} and \textit{heterogeneous reactive} dynamics, which makes the development of flexible control solutions a challenging problem. In this context, we propose Implicit Control, an approach that leverages numerical analysis theory to find suitable herding inputs even when the nonlinearities in the evaders' dynamics yield to \textit{implicit equations}. 
The intuition behind this methodology consists in driving the input, rather than computing it, towards the \textit{unknown} value that achieves the desired dynamic behavior of the herd. The same idea is exploited to develop an adaptation law, with stability guarantees, that copes with uncertainties in the herd's models.
Moreover, our solution is completed with a novel caging technique based on uncertainty models and Control Barrier Functions (CBFs), together with a distributed estimator to overcome the need of complete perfect measurements. Different simulations and experiments validate the generality and flexibility of the proposal.
\end{abstract}

\IEEEpeerreviewmaketitle


\section{Introduction}\label{sec:intro}

\IEEEPARstart{R}{ecent} advances in Multi-Robot Systems (MRS) have favored the development of successful control strategies in real-life problems such as entrapment~\cite{Antonelli_RAM_2008_Entrapment}, hunting~\cite{Zhu_IJARS_2015_Hunting} or escorting~\cite{Gao_ACCESS_2018_Escorting}. Despite the different nature of scenarios, these problems can all be seen as different instances of \textit{herding}~\cite{Pierson_2018_TR_Herding}, where the objective is to drive a group of targets or \textit{evaders} to specific locations using a team of robots or \textit{herders}. Common to all of them is the non-cooperative nature of the evaders with respect to the control objective, typically entangled in complex nonlinear and heterogeneous behaviors. In fact, the difficulties hidden in the herding problem have motivated broad interdisciplinary research assembling physiologists, mathematicians and neurologists with engineers~\cite{nolfi2002power,Strombom2014Shepherding,Long2020Comprehensive}.  

\begin{figure}[!ht]
\centering
\begin{tabular}{cc}
     \includegraphics[width=0.48\columnwidth,height=0.3\columnwidth]{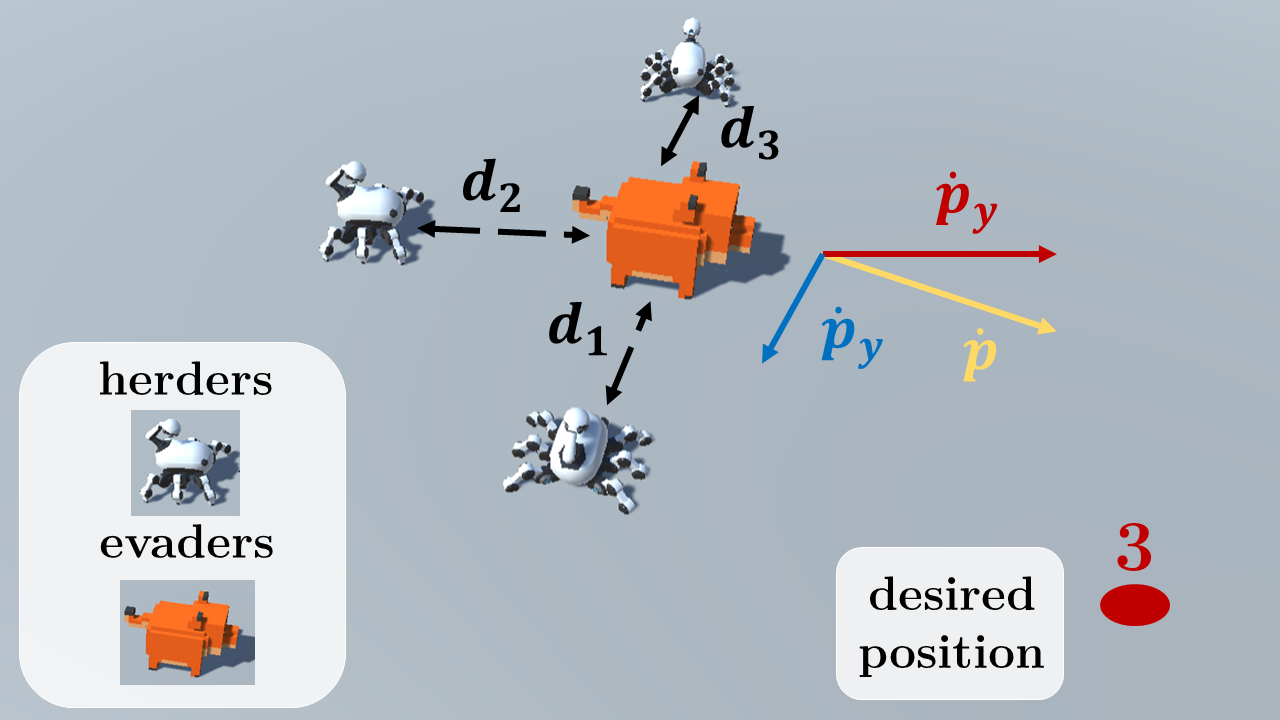}\hspace{0.1cm}
     & 
     \includegraphics[width=0.48\columnwidth,height=0.3\columnwidth]{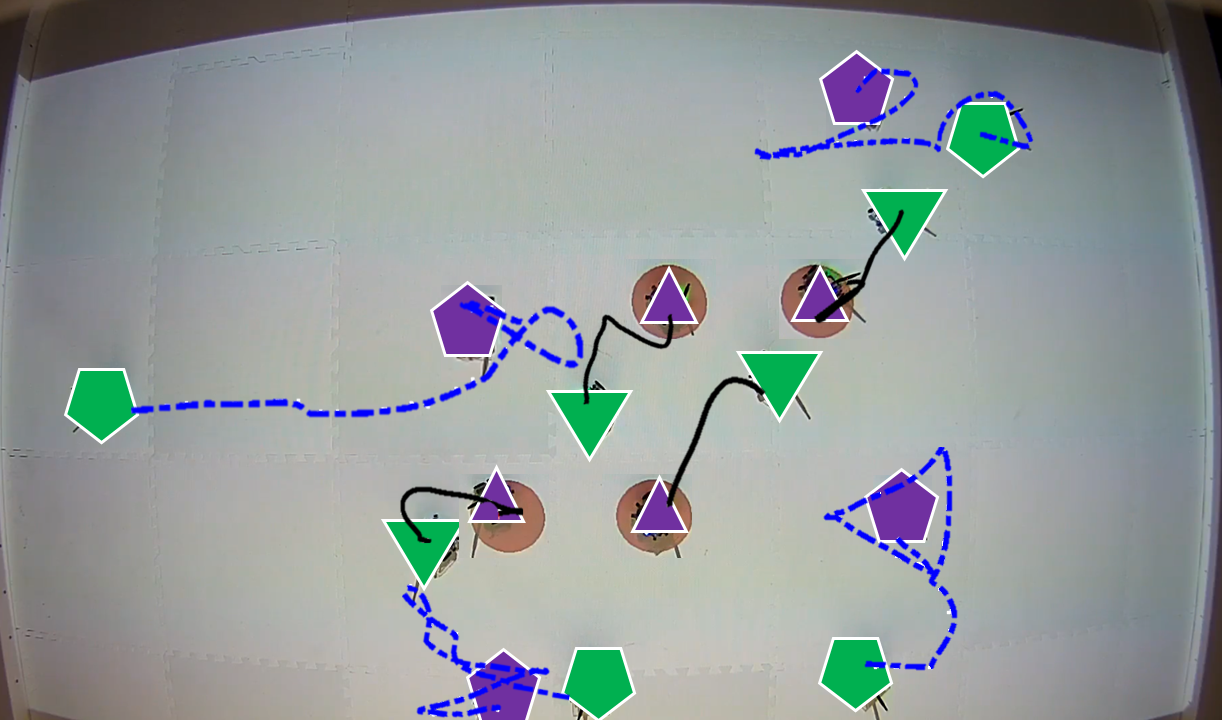}\hspace{0.1cm}
\end{tabular}
    
\caption{In herding, the robotic herders drive the evaders towards specific locations simultaneously. The control strategy exploits repulsive forces (left) to place the evaders in their desired positions. The experiment in the right depicts four robots herding four evaders. The symbols are explained in Table~\ref{table:legend}. The supplementary video includes the complete execution of the experiment.}
\label{fig:first_impression}
\end{figure}

To cope with these difficulties, we present a novel control solution, called \textit{Implicit Control}, that is capable of steering a group of evaders towards \textit{individually} assigned goals \textit{simultaneously}, irrespective of their heterogeneous and non-cooperative nonlinear dynamics.
The proposal leverages numerical analysis theory to derive proofs of existence and stability of the control. Furthermore, these fundamentals are used to develop a general adaptation law that tackles the herding in the presence of uncertainties in the evaders' models. With the inclusion of a distributed estimator and a caging strategy we present a complete, flexible, general herding solution. The flexibility and generality of Implicit Control also allows to steer large herds to a common region or to control different herds individually.

In more detail, the \textit{contributions} of this work are:
\begin{itemize}
    \item A novel control technique, called Implicit Control (Section~\ref{sec:IC}), that solves a precision herding problem in the context of MRS. Implicit Control is able to impose desired dynamics for general input-nonaffine systems. This is done even if some of the parameters are unknown thanks to an adaptation law (Section~\ref{sec:adaptation}) that is derived from the principles of Implicit Control. 
    \item A distributed estimator and a caging technique based on uncertainty models and CBFs (Section~\ref{sec:distributed}) that soften some requirements of the controller. The distributed estimator is a derived design from the Implicit Control: it considers, explicitly, the control proposal to enhance the estimation; besides, it accounts for partial measurements, overcoming the need of complete perfect measurements. With the caging technique the herders encircle the evaders, approaching and surrounding the evaders before they can escape. Simulations (Section~\ref{sec:simulations}) and experiments (Section~\ref{sec:experiments}) with real robots validate the success of the proposal.
\end{itemize}

This paper is an \textit{evolved version} of~\cite{sebastian2020multi}. Compared to it, the theoretical development of Implicit Control is more comprehensive, including a stability analysis for general systems that was not discussed in~\cite{sebastian2020multi}.
The ability to deal with uncertain models using an adaptation law is also new. 
In addition, we formulate an extended version of the Distributed Kalman Filter (E-DKF)~\cite{olfati2007distributed} to overcome the need of complete perfect measurements. 
The filter embeds the Implicit Control strategy in the prediction stage and works with input-dependent measurements, two properties that are atypical and innovative in this context.
Furthermore, the caging technique combining uncertainty models and CBFs is also novel.  We have added new simulated experiments, where we validate the caging technique and the distributed estimator. Furthermore, it is shown that Implicit Control can be exploited to herd several evaders with a few herders by simply changing the control objective.
Finally, we include new experimental results that validate the herding solution. 


\section{Related work}\label{sec:related}

To discuss the related work, we distinguish between (i) single-agent or multi-agent herding solutions, (ii) herding of evaders with linear or nonlinear dynamics, and (iii) herding with known or unknown evaders' dynamics.

Most of the existing works deal with the herding of a \textit{single entity}. An early example is~\cite{Antonelli_RAM_2008_Entrapment}, which employs a Null space-based behavioral control to escort/entrap the target. 
The authors of~\cite{Zhu_IJARS_2015_Hunting} apply a Bio-inspired Neural Network to hunt in an underwater environment. 
To produce an escorting behavior,~\cite{du2017pursuing} develops a distributed switching strategy, where the escorting task is assigned from one herder to another when a Voronoi boundary is crossed. Escorting shares properties with the encirclement in~\cite{franchi2016decentralized}, where decentralized geometrical strategies are derived depending on the evaders' configuration. 

Few papers solve the herding of \textit{more than one evader}. In~\cite{Jahn_ICRA_2017_Surveillance}, a group of robots navigates around a certain area to prevent evaders to cross its boundaries. The work in~\cite{Pierson_2018_TR_Herding} drives groups of entities by an active encirclement but does not consider specific final positions for each evader. Following a similar approach, the authors in~\cite{Panagou_CDC_2019_Herding} go a step further and propose an active encirclement which avoids obstacles. With the same spirit,~\cite{Stilwell_TSMC_2005_Redundant} and~\cite{Antonelli_TRO_2006_Kinematic} develop solutions to drive a whole herd towards a certain region, where it is required that some global features of the herd (e.g., mean position) converge to a desired equilibrium. Besides, both works deal with self-controllable entities, i.e., the controller, state to control, and input are all embedded in the same entity. Recently,~\cite{Auletta_AuRo_2022_Herding} exploits simple local rules to push all herders towards the same containment area. In~\cite{Licitra_CSL_2018_Herding}, a single robot is in charge of herding multiple evaders by controlling them one by one, whereas~\cite{Licitra_TRO_2019_Single} extends the solution with a Neural Network that learns the dynamics of the targets. The herding in~\cite{elamvazhuthi2020controllability} considers a single robot to control a herd as well, but in this case the evaders cooperate with the robotic leader to follow it. The authors of~\cite{song2021herding} present an approach based on algebraic topology called ``herding by caging'', where repulsive forces are leveraged to steer the evaders, again, towards a certain general region. In contrast, our method performs the herding of \textit{all} the evaders to \textit{precise individual locations simultaneously}. Additionally, our approach can also be used to steer large herds towards a global region or to split and steer each sub-herd towards individually assigned regions simultaneously.

Another feature of herding problems is the behavior of the evaders, where two assumptions are often considered: \textit{linear and homogeneous dynamics}. The first assumption is considered in, e.g.,~\cite{Jiang_TAC_2019_Containment}, where the containment of linear heterogeneous agents is performed by a time-varying formation.  
The work in~\cite{Anisi_TAC_2010_Surveillance} drives a group of linear UGVs to ensure that a region is completely surveyed whilst in~\cite{Ramana_CDC_2015_PE} both pursuers and evader have the same linear dynamics, solving the problem with geometric tools. The second assumption is used in~\cite{Pierson_2018_TR_Herding}, where agents with nonlinear homogeneous repulsive dynamics are controlled using a team of robots. A different instance is~\cite{Alexopoulos_ICUAS_2017_PE}, where a complete control structure is presented, from the hex-rotor motion to the pursuit layer. Our proposal offers a general framework to design a control strategy for the herding of \textit{general nonlinear and heterogeneous} evaders.

Most of the aforementioned solutions assume perfect knowledge of the herd dynamics. To deal with uncertainty, a common approach is the design of model-free strategies. The work in~\cite{Franchi_IFAC_2010_Encirclement} localizes the target with relative-position sensors and then applies probabilistic tools for the encirclement. Similarly,~\cite{Zhu_ICPR_2018_PE} includes information from a camera to learn the behavior of the target. 
These instances deal with the herding of a single target. An extension to $n$ robots detecting $n$ targets can be found in~\cite{Desouky_ICSMC_2010_PE2}, using a mixture of learning, optimization and fuzzy techniques. The use of Gene Regulatory Networks is studied in~\cite{Peng_IJARS_2016_Trapping} to entrap groups of targets.
In both cases, no accuracy in the final position of the evaders is required, and in~\cite{Peng_IJARS_2016_Trapping} the targets do not react against the robots' efforts. 
This lack of accuracy and flexibility in the number of evaders under uncertainty is solved in our proposal by an \textit{adaptive technique}. Thus, model information is exploited for the precise herding of repulsive heterogeneous evaders while uncertainty is tackled by adapting online the unknown parameters.

It is not the first time that the term ``implicit control'' appears in the literature, but this is never done in terms of a novel control technique, unlike our coined Implicit Control. In~\cite{Estrada_CDC_1996_Implicit}, the authors deal with a time-varying SISO system by first designing a trivial time-invariant controller and then reformulating it to account for unknown parameters. Therefore, a classical linear control law is reshaped to \textit{implicitly} operate irrespective of the time-varying parameters. Meanwhile, in~\cite{Sueoka_ROBIO_2012_Implicit} and~\cite{Osuka_IROS_2010_Implicit} the term ``Implicit Control'' is inspired by the fast adaptation that emerges in living entities. In such systems, the frontiers between plant and controller are fuzzy. They refer to ``Implicit Control Law'' as that abstract intersection. In contrast, our ``Implicit Control'' is a novel control technique that considers a general \textit{implicit equation} to describe the control objective, and then develops an expression for the time-derivative of the input such that the implicit equation converges towards its solution. That convergence \textit{implicitly} implies that the state converges to the desired equilibrium, with the desired properties encoded in the \textit{implicit equation}.

Lastly, regarding the control technique we propose in this paper, recent works in the control theory community support the advances in developing solutions where the control expression is in terms of the dynamics of the input instead of the classical algebraic expression. In~\cite{Blanchini2017ModelFree}, the time derivative of the input is used to characterize and tune static plants. However, to the best of our knowledge, this work has never been applied in dynamical systems to build an explicit controller. Meanwhile,~\cite{shivam2019tracking} shows the advantages of this novel control perspective in applications with highly nonlinear input-nonaffine dynamics, with an input determined by the ``flow'' that a Newton-Raphson solver will follow to arrive at the roots of a nonlinear equation. Another example is~\cite{wardi2017performance}, where a similar approach is applied in a tracking problem. Our proposal, sustained in~\cite{sebastian2020multi}, goes a step further in generality. We give formal proofs of existence and stability that are then corroborated in both simulations and real experiments, solving a problem where, unlike the aforementioned papers, the control objective is evasive with respect to the inputs. In addition, the dynamics to be controlled are partially unknown, which is also a distinguishing factor.


\section{Problem Statement}\label{sec:prosta}

We consider the problem of herding a group of $m$ evaders using a team of $n$ robotic herders. We denote the evaders by $\textit{j} \in \{ 1,  ... , m \},$ and the herders by $\textit{i} \in \{1, ... , n\}$. Since we aim at using the position of the herders to control the position of the evaders, we define the state $\mathbf{x} \in X \subseteq \mathbb{R}^{2m}$ and the input $\mathbf{u} \in U \subseteq \mathbb{R}^{2n}$ as $\mathbf{x} = \begin{bmatrix}
    \mathbf{x}_{1}^{T} & \dots  & \mathbf{x}_{m}^{T}
    \end{bmatrix}^{T}$ and $\mathbf{u} = \begin{bmatrix}
    \mathbf{u}_{1}^{T} & \dots  & \mathbf{u}_{n}^{T}
    \end{bmatrix}^{T},$ where $\mathbf{x}_{j} = [x_j^e,y_j^e] \in X_j \subseteq \mathbb{R}^{2}$ is the position of evader $j$ and $\mathbf{u}_{i} = [x_i^h,y_i^h] \in U_i \subseteq \mathbb{R}^{2}$ is the position of herder $i$, operating in a 2D space\footnote{The choice of a 2D space is to fit the real experiments, but the solution can be generalized to higher dimensions without changes in the formulation.}. The movement of each evader is described by general dynamics
\begin{equation}
\dot{\mathbf{x}}_j = f_j(\mathbf{x},\mathbf{u}), 
\label{eq:base}
\end{equation}
allowing for any nonlinear behavior encoded in $f_j(\mathbf{x},\mathbf{u})$. Notice that we account for the possibility of nonlinear dependencies on the input. The only assumption regarding $f_j$ is that it is of class $C^1$ for all $j$. We exemplify Eq.~\eqref{eq:base} using two dynamic models from the literature. Their choice is motivated by the strongly nonlinear behavior in the position of evaders and herders, described by input-nonaffine dynamics.

The \textit{Inverse Model} (adapted from \cite{Pierson_2018_TR_Herding}) is
\begin{equation}
\dot{\mathbf{x}}_j = f_j^{inv}(\mathbf{x},\mathbf{u}) = \theta_j\sum_{i=1}^{n} \frac{\mathbf{d}_{ij}}{||\mathbf{d}_{ij}||^{3}}
\label{eq:PiersonBase}
\end{equation}
where $\mathbf{d}_{ij} = \mathbf{x}_{j} - \mathbf{u}_{i}$ is the relative position between evader $j$ and herder $i$, and $\theta_j$ is a positive constant which expresses the aggressiveness in the repulsion provoked by the herders. 
Note that despite the model has a singularity in $\mathbf{d}_{ij}=\mathbf{0}$, in practice the speed remains bounded. The repulsion grows with $\frac{1}{||\mathbf{d}_{ij}||^{3}}$ so the closer the herders, the larger the repulsion. The only way of achieving $\mathbf{d}_{ij}=\mathbf{0}$ is that a perfectly evenly distributed number of herders approaches the evader. This does not happen in practice because herders need to impose non-zero repulsive forces to steer the evaders.

The \textit{Exponential Model} (adapted from \cite{Licitra_CSL_2018_Herding}) is
\begin{equation}
\label{eq:LicitraBase}
\begin{aligned}
\dot{\mathbf{x}}_j =& f_j^{exp}(\mathbf{x},\mathbf{u}) =
\\&
\theta_j\sum_{i=1}^{n}\mathbf{d}_{ij}e^{-\chi_{ij}}(1-\beta_j\hbox{sigm}(-||\mathbf{d}_{ij}|| + d_{\min})) 
\end{aligned}
\end{equation}
where $\chi_{ij} = \frac{1}{\sigma_j^2}\mathbf{d}_{ij}^{T}\mathbf{d}_{ij}$ and $\sigma_j>1$. In this model there is a switching condition if $||\mathbf{d}_{ij}|| \leq d_{\min}$, where the evader $j$ becomes ``scared'' when the distance with some herder $i$ is smaller than $d_{\min}$ and the intensity of the repulsive interaction increases, due to $0<\beta_j<1$. In practice, the sigmoid function $\hbox{sigm}(\cdot)$ is included in~\eqref{eq:LicitraBase} to avoid the violation of the $C^1$ assumption. In a similar fashion, a sigmoid function or an
hyperbolic tangent can be used to model saturations in the speed of the evaders.

Given the individual dynamics of the evaders in Eq.~\eqref{eq:base}, the joint system dynamics can be defined as
\begin{equation}
\dot{\mathbf{x}} = f(\mathbf{x},\mathbf{u})
\label{eq:initial_system},
\end{equation}
where $f(\mathbf{x},\mathbf{u})$ simply comes from stacking all $f_j(\mathbf{x},\mathbf{u})$.
This formulation allows to consider heterogeneous herds, with different number of evaders and motion models. 

Our goal is to herd the evaders towards the desired positions $\mathbf{x}^{*}\in X$ simultaneously. To do this, we define the position error of the evaders as $\widetilde{\mathbf{x}} = \mathbf{x} - \mathbf{x}^{*}$, and we set the control objective to be to drive $\widetilde{\mathbf{x}}$ to zero. This is what we call \textit{precise herding}. A particular instance of precise herding is the herding of a herd towards the same region, where some of the elements of $\mathbf{x}^*$ are specified by the practitioner and the rest of them are free or directly dependent on the others. Nevertheless, in general, we consider the case where each evader $\mathbf{x}_j$ has an individually assigned $\mathbf{x}_j^*$.

It is noteworthy that the reactive behavior of the evaders is with respect to the position of the herders. Therefore, a control strategy which determines inputs in terms of herders' positions is adequate to generalize the solution to different robotic platforms. This high-level approach can be combined with any robot-specific low-level controller. 

To keep the generality of the solution, in this work we assume that the maximum velocity of both herders and evaders is $v_{\max}$. This requires an initial caging phase, surrounding the evaders to avoid their escape before the precision herding begins. It is not mandatory to achieve a compact and closed encirclement to succeed but just the distribution of the herders near the evaders. This is the reason of developing a caging stage to complete the herding, described in Subsection~\ref{subsec:approaching}.


\section{Implicit Control}\label{sec:IC}

Herding seeks an expression for the input $\mathbf{u}$ such that the evaders go to their desired positions $\mathbf{x}^*.$
Besides, the herding may need to accomplish other requirements, such as a desired transient response. This can be translated into designing $\mathbf{u}$ such that the evaders follow a desired dynamics $f^*$,
\begin{equation}\label{eq:def_f_star}
    \dot{\mathbf{x}} = f^*(\mathbf{x}).
\end{equation}
To ease the analysis we consider, without loss of generality, $\mathbf{x}^*=\mathbf{0}$ as the desired equilibrium point of the system. The function $f^*$ can adopt any desired structure\footnote{The choice of $f^*$ as a state-dependent function is to follow the typical desired dynamics in regulation problems. However, the results can be directly extended to desired dynamics that depend on the input as well.} with the constraint of being stable in the equilibrium point. 

In order to control the evaders, the typical practice consists in finding a closed expression for $\mathbf{u}$ that depends only on the state and maps~\eqref{eq:initial_system} into~\eqref{eq:def_f_star}. Unfortunately, due to the complexity of the system, sometimes it is extremely challenging and time consuming to find such explicit mapping. Indeed, trying to find it generally yields to an implicit expression for the input. The analytical solution of the latter is not guaranteed, and the mechanisms to calculate such input do not certify that the desired properties prevail anymore. In contrast, we propose Implicit Control as a general control solution to achieve desired closed-loop dynamics despite the complexity of~\eqref{eq:initial_system}, overcoming the necessity of looking for a controller to perform the mapping from~\eqref{eq:initial_system} to~\eqref{eq:def_f_star}.

We study the conditions that allow to find $\mathbf{u}$ such that the evaders evolve according to $f^*$. Then, we propose a design procedure to solve the control input. The description is kept in general control terms since we believe that this procedure can be of interest in other control problems. 

\subsection{Control existence}\label{subsec:existence}

Firstly, it is necessary to know if there exists a smooth input that makes the actual dynamics equal to the desired ones. For convenience, let define
\begin{equation}\label{eq:def_h_CDC}
    h(\mathbf{x},\mathbf{u}) = f(\mathbf{x},\mathbf{u}) - f^*(\mathbf{x})
\end{equation}
as the working equation, shifting the analysis to that of looking for the existence of $\mathbf{u}$ such that $ h(\mathbf{x},\mathbf{u}) = \mathbf{0}$, i.e., looking for the existence of the \textit{roots} of $h(\mathbf{x},\mathbf{u})$. 

At this point, a straightforward solution can be to compute the input using a numerical method to find the roots of $h$. This is a simple and powerful approach, but lacks of formal guarantees, depends on the numerical method and is computationally expensive since we must ensure that the numerical method finds the roots of $h$ at all instants. Therefore, we will use this approach as a baseline to compare the Implicit Control solution.

In any case, the first step is to determine the existence of the roots of $h$, which is, in general, not trivial. We review some fundamentals of numerical analysis theory to show sufficient conditions to ensure existence and smoothness\footnote{Smoothness is not mandatory, but it is a desirable property.} of $\mathbf{u}$. The Implicit Function Theorem (Theorem 9.28 of~\cite{Rudin1976Mates}), applied to $h$, offers a set of sufficient conditions to guarantee the existence and smoothness of $\mathbf{u}$ in a local region.

\begin{theorem}[Adapted from Theorem 9.28 of~\cite{Rudin1976Mates}]
\label{theorem:exist_and_smooth}
Let 
\begin{equation}\label{eq:formal_h}
    h: I = X \times U \subset\mathbb{R}^{M}\times \mathbb{R}^{N}\longmapsto\mathbb{R}^{N}
\end{equation}
a $C^1$-mapping, such that $h(\mathbf{x}^*,\mathbf{u}_0^{*})=\mathbf{0}$ for some point $\mathbf{u}_0^{*} \in U$. Additionally, consider the Jacobian
\begin{equation}\label{eq:J_in_depth}
    \mathbf{J} \kern -2pt = \kern -2pt \left(\mathbf{J}_{\mathbf{x}} | \mathbf{J}_{\mathbf{u}})\right. \kern -3pt = 
    \kern -4pt \left.\begin{pmatrix}
    \frac{\partial h_1}{\partial \mathbf{x}_1} & \hdots & \frac{\partial h_1}{\partial \mathbf{x}_M}
    & 
    \frac{\partial h_1}{\partial \mathbf{u}_1} & \hdots & \frac{\partial h_1}{\partial \mathbf{u}_N}
    \\
    \vdots & \ddots & \vdots & \vdots & \ddots & \vdots
    \\
    \frac{\partial h_M}{\partial \mathbf{x}_1} & \hdots & \frac{\partial h_M}{\partial \mathbf{x}_M}
    &
    \frac{\partial h_M}{\partial \mathbf{u}_1} & \hdots & \frac{\partial h_M}{\partial \mathbf{u}_N}
    \end{pmatrix}\right.\kern -2pt.
\end{equation}
such that $\mathbf{J}_{\mathbf{u}}$ is non-singular in the point $(\mathbf{x}^*,\mathbf{u}_0^{*})$. Then, there exist open subsets $I^*\subset\mathbb{R}^{M}\times \mathbb{R}^{N}$ and $X^*\subset\mathbb{R}^{M}$, with  $(\mathbf{x},\mathbf{u}^{*}) \in I^*$ and $\mathbf{x}\in X^*$, having the following property: to every possible $\mathbf{x} \in X^*$ corresponds a unique $\mathbf{u}^*$ such that $(\mathbf{x},\mathbf{u}^{*}) \in I^*$ and $h(\mathbf{x},\mathbf{u}^*)=\mathbf{0}$.
\end{theorem}

The Theorem imposes three conditions to be fulfilled.
Firstly, there must exist an input, $\mathbf{u}_0^{*},$ which solves the control in $\mathbf{x}^{*}$. In the herding context, there must exist a stable configuration of the herders when the evaders are in their desired positions. A general condition to ensure this, for all $\mathbf{x} \in X$, is to have at least the same number of inputs than states to control. Otherwise, it is not always possible to control the system because there are less degrees of freedom than states to control. In the case of the precision herding problem, this means $n \geq m,$ since both evaders and herders are first order entities in the space and we want to \textit{control all the evaders simultaneously;} i.e., at all instants, all the herders contribute to the motion of all the evaders. Besides, a team of herders is needed so that, at equilibrium, the total repulsive force in each evader is zero.
Nevertheless, we demonstrate in Section~\ref{sec:simulations} that Implicit Control can also be used to control large-scale herds with only a few robotic herders if the control objective is the position of the centroid of the herd.

Secondly, $h$ must be of class $C^1$ in $(\mathbf{x}^*,\mathbf{u}_0^{*})$. If $f^*$ is chosen of class $C^1$ in $\mathbf{x}^*$, then the condition is accomplished because $f_j$ in~\eqref{eq:base} is of class $C^1$ for all $j$, so $f$ in~\eqref{eq:initial_system} is of class $C^1$.

The last condition requires the Jacobian of $h$ with respect to $\mathbf{u}$, $\mathbf{J}_{\mathbf{u}}$, to be non-singular in the desired location. Since for $m\neq n$ the matrix is not square, the right Moore-Penrose inverse matrix $\mathbf{J}_{\mathbf{u}}^{+}  = \mathbf{J}_{\mathbf{u}}^T(\mathbf{J}_{\mathbf{u}} \mathbf{J}_{\mathbf{u}}^T)^{-1}$ is generally considered as the one to be non-singular. An alternative is to resort to the Constant Rank Theorem (Theorem 5.22 in~\cite{lee2003smooth}) and find a projection of lower dimension which locally preserves the properties of the Jacobians but achieving a square matrix. The use of the pseudoinverse can be seen as this projection. Given the aforementioned features of $h$ and $\mathbf{J}_{\mathbf{u}}$, the last condition is accomplished in $\mathbf{x}^*$. Moreover, by restricting $I$ to the subspace without collisions it is ensured that the two last conditions of Theorem~\ref{theorem:exist_and_smooth} hold for all $\mathbf{x}$ and, therefore, the Theorem holds for all $\mathbf{x}$ in this subspace. Considering that each herder provokes a repulsive reaction in every evader, collision among herders and evaders will not happen. Similarly, since herders are the controlled entities, it is easy to prevent collisions among them.

Theorem~\ref{theorem:exist_and_smooth} is an analytical tool to guarantee the necessary conditions to afford the Implicit Control proposal. Further analysis on the existence of input will depend on the particularities of the controlled system. Therefore, we consider the following assumption, so that a suitable input can be computed.

\begin{assumption}\label{assumption:conditions}
There exists $\mathbf{u}_0$ such that $h(\mathbf{x},\mathbf{u}_0)=\mathbf{0}$, the functions $f$ and $f^*$ are of class $C^1$ and $\mathbf{J}_{\mathbf{u}}$ has rank $m$ for all $(\mathbf{x}, \mathbf{u}) \in I$.
\end{assumption}

The problems of uniqueness of solution or existence of local minima are not a concern since no cost function is being optimized. We underline this to emphasize that Implicit Control does not search for any local or global optimum, but for the roots of $h$.

\subsection{Control design}\label{subsec:design}

The proposed design method consists in expanding the initial system in~\eqref{eq:initial_system} with an input dynamics that converges to the roots of $h$. This transforms the problem to that of computing the input $\mathbf{u}$ as part of an expanded explicit system, described in continuous time and with analytical solution. To ease the presentation, we first assume $n=m$, generalizing to $n\neq m$ later.

The function $h$ is determined by the desired closed-loop behavior $f^*$. In an ideal scenario, $h=0$ always holds. However, achieving $h=0$ for all time is generally not an easy task. Our solution consists in considering $h$ as a dynamic system. The evolution of $h$ over time is not determined a priori; therefore, the input dynamics can be designed such that the evolution of $h$ follows a desired dynamics $h^*$,
\begin{equation}\label{eq:def_h_star}
    \frac{d h(\mathbf{x},\mathbf{u})}{d t} = h^*(\mathbf{x},\mathbf{u}),
\end{equation}
so that $h$ converges to zero. From now on, we assume that $h^*$ is chosen to make $h$ stable. With this in mind, we propose the following expression for the input dynamics,
\begin{equation}
\label{eq:u_dynamics}
\kern -4pt
    \dot{\mathbf{u}} = \mathbf{J}_{\mathbf{u}}^{-1}
    \kern -3pt\left(h^*(\mathbf{x},\mathbf{u})-\mathbf{J}_{\mathbf{x}} f(\mathbf{x}, \mathbf{u}) \right),
\end{equation}
where $\mathbf{J}_{\mathbf{x}}$ is the Jacobian of $h$ with respect to $\mathbf{x}$ and $\mathbf{J}_{\mathbf{u}}$ is the Jacobian of $h$ with respect to $\mathbf{u}$. 
This allows to build an explicit expanded system of the form
\begin{equation}
\label{eq:diff_eq_system_u}
\left\{
\begin{aligned}
    \dot{\mathbf{x}} &= f(\mathbf{x}, \mathbf{u}) 
    \\ 
    \dot{\mathbf{u}} &= \mathbf{J}_{\mathbf{u}}^{-1}
    \kern -3pt\left(h^*(\mathbf{x},\mathbf{u})-\mathbf{J}_{\mathbf{x}} f(\mathbf{x}, \mathbf{u}) \right)
\end{aligned}
\right. .
\end{equation}
The control problem is then reduced to analyzing the stability of the system in~\eqref{eq:diff_eq_system_u}. With the proposed expansion the complete system becomes autonomous, i.e., the input becomes part of the state of the expanded system. As the structure in~\eqref{eq:diff_eq_system_u} is very general, we provide stability results for typical dynamics. The first one considers that the original system is Input-to-State Stable (ISS). 
\begin{theorem}
\label{Theorem:ISS}
If the system defined in~\eqref{eq:initial_system} is ISS and Assumption~\ref{assumption:conditions} holds, then the input dynamics in~\eqref{eq:u_dynamics} ensures the convergence of $h$ towards zero and achieves the desired closed-loop dynamics in~\eqref{eq:def_f_star}.
\end{theorem}
\begin{proof}
Applying the chain rule over $dh/dt$ yields to
\begin{equation}\label{eq:der_h}
    \frac{d h(\mathbf{x},\mathbf{u})}{d t} =  \mathbf{J}_{\mathbf{x}}(\mathbf{x},\mathbf{u}) \frac{d \mathbf{x}}{d t} + \mathbf{J}_{\mathbf{u}}(\mathbf{x},\mathbf{u}) \frac{d \mathbf{u}}{d t},
\end{equation}
where $\mathbf{J}_{\mathbf{x}}$ and $\mathbf{J}_{\mathbf{u}}$ are the Jacobians of $h$ with respect to $\mathbf{x}$ and $\mathbf{u}$. By Assumption~\ref{assumption:conditions}, $\mathbf{J}_{\mathbf{u}}$ is full rank. Then, the substitution of the input dynamics in~\eqref{eq:u_dynamics} gives~\eqref{eq:def_h_star}, meaning that $\dot{\mathbf{u}}$ imposes the desired dynamics over $h$. 

This ensures that, at some point, $h = \mathbf{0}$. However, in the transient, $h \neq \mathbf{0}$ so the input $\mathbf{u}$ is not the one that achieves $f=f^*$. This non-zero difference can be seen as a perturbation in the input $\mathbf{u}$ of system~\eqref{eq:initial_system}. Nevertheless, since the system is ISS, the stability remains despite the perturbation~\cite{Khalil2014Nonlinear}. Due to the stability of $h$, the perturbation vanishes with time, and the expanded system in~\eqref{eq:diff_eq_system_u} makes~\eqref{eq:initial_system} to converge to~\eqref{eq:def_f_star}.
\end{proof}

The ISS property is common but particular to some systems. The next case of study aims at using Implicit Control for more general dynamics. First, we restrict the possible desired closed-loop behaviors to
\begin{equation}
\label{Eq:f_star_real}
f^*(\mathbf{x})=-\mathbf{K}_f(\mathbf{x})\mathbf{x},
\end{equation}
where $\mathbf{K}_f(\mathbf{x})$ is a function which returns a positive definite matrix. In particular, $\mathbf{K}_f(\mathbf{x})$ can be a polynomial whose terms are all of even order. For instance, if $\mathbf{K}_f(\mathbf{x}) = \mathbf{K}_f$, then  $f^*(\mathbf{x})= -\mathbf{K}_f \mathbf{x}$ has a linear behavior with its corresponding exponential-like transient, desired settling time and stability properties. 
Similarly, let
\begin{equation}
\label{eq:h_star}
    h^*(\mathbf{x},\mathbf{u}) =
    - \mathbf{K}_h (h(\mathbf{x},\mathbf{u}))h(\mathbf{x},\mathbf{u}),
\end{equation}
with $\mathbf{K}_h(h(\mathbf{x},\mathbf{u}))$ a function which returns a positive definite matrix as well. 

\begin{theorem}
\label{Theorem:u_dynamics}
Given Assumption~\ref{assumption:conditions}, the definition of $f^*$ in~\eqref{Eq:f_star_real} and the definition of $h^*$ in~\eqref{eq:h_star}, if matrices $\mathbf{K}_f(\mathbf{x})$ and $\mathbf{K}_h(h(\mathbf{x},\mathbf{u}))$ are chosen in such a way that
\begin{equation}\label{eq:negative_matrix}
    \mathbf{K} = \left.\begin{pmatrix}
    -\mathbf{K}_f(\mathbf{x}) & 0.5 \mathbf{I}
    \\
    0.5 \mathbf{I} & -\mathbf{K}_h (h(\mathbf{x},\mathbf{u}))
    \end{pmatrix}\right.
\end{equation}
is negative definite, then the system in Eq.~\eqref{eq:diff_eq_system_u} is Globally Asymptotically Stable (GAS).
\end{theorem}
\begin{proof}
Consider $V = \frac{1}{2}\mathbf{x}^T\mathbf{x} + \frac{1}{2}h^T h$ as a Lyapunov candidate function, whose derivative is
\begin{equation}
    \dot{V} = \mathbf{x}^T \dot{\mathbf{x}} + h^T \dot{h} = \mathbf{x}^T f + h^T h^*
\end{equation}
and where $\dot{h}$ is substituted by~\eqref{eq:def_h_star} following the same steps of the proof of Theorem~\ref{Theorem:ISS}. 
Using the definition of $h$ in~\eqref{eq:def_h_CDC} yields to
\begin{equation}
    \dot{V} = \mathbf{x}^T h + \mathbf{x}^T f^* + h^T h^*.
\end{equation}
Considering $f^*$ and $h^*$ from Eqs.~\eqref{Eq:f_star_real} and~\eqref{eq:h_star}, $\dot{V}$ is
\begin{equation}\label{eq:Ly_theorem_der_3}
    \dot{V} = \begin{pmatrix} \mathbf{x}^T & h^T \end{pmatrix} \mathbf{K} \begin{pmatrix} \mathbf{x}^T & h^T \end{pmatrix}^T.
\end{equation}
Then, designing $\mathbf{K}_f(\mathbf{x})$ and $\mathbf{K}_h(h(\mathbf{x},\mathbf{u}))$ such that $\mathbf{K}$ is negative definite guarantees that~\eqref{eq:diff_eq_system_u} is GAS, and both $h$ and $\mathbf{x}$ go to zero.
\end{proof}

In practice, $||\mathbf{K}_h|| \gg ||\mathbf{K}_f||$ imposes the convergence of $h$ to be much faster than the desired closed-loop dynamics, so the evaders will behave following $f^*$.

\subsection{Control extensions}\label{subsec:extensions}

Assuming $n=m$ is a restrictive condition. 
We deal with the generalization to $n \neq m$ by 
replacing $\mathbf{J}_{\mathbf{u}}^{-1}$ with the pseudoinverse $\mathbf{J}_{\mathbf{u}}^{+}$. In the case that there are some states whose dynamics are not explicitly dependent on $\mathbf{u}$ (for instance, the velocity of the evaders might not depend on herders' position, but their acceleration), then a 
linear transformation $\mathbf{T}$ over $h$ can be applied
\begin{equation}\label{eq:transformation}
    \bar{h}(\bar{\mathbf{x}},\mathbf{u}) = \mathbf{T}h(\mathbf{x},\mathbf{u}).
\end{equation}
The purpose is to create a new group of states $\bar{\mathbf{x}}$ combining all the states which do not depend on the input with some of the states which do depend on the input. Then, Implicit Control can be applied, over the new function $\bar{h}$, to control states which do not explicitly depend on the input.

Furthermore, all the previous results can be extended to dynamic references. Let consider a desired state reference $\mathbf{x}^*$ with desired dynamics $\dot{\mathbf{x}}^*$. 
\begin{corollary}\label{corollary}
Assuming that $\dot{\mathbf{x}}^*$ is smooth and bounded, the results in Theorems~\ref{Theorem:ISS}-\ref{Theorem:u_dynamics} hold if the function $h$ is redefined as
\begin{equation}\label{eq:TV_h}
    h(\mathbf{x},\mathbf{u}) =  f(\mathbf{x},\mathbf{u})-f^{*}(\mathbf{x})- \dot{\mathbf{x}}^{*}.
\end{equation}
\end{corollary}
\begin{proof}
By assuming that $\dot{\mathbf{x}}^*$ is smooth and bounded we maintain the existence and smoothness of the input. The original system is still described by~\eqref{eq:initial_system} but the control objective changes to
\begin{equation}\label{eq:TV_change}
    f^*(\mathbf{x},\mathbf{u}) = f(\mathbf{x},\mathbf{u}) - \dot{\mathbf{x}}^{*}.
\end{equation}
The expression comes from imposing that the error dynamics follow $f^*$. This information can be incorporated in the design of the function $h$ by
adding the known time-varying term in~\eqref{eq:TV_change} to compensate the dynamic reference. The time derivative of $\dot{\mathbf{x}}^*$ is implicitly included in the Jacobians of the new $h$, therefore the structure in system~\eqref{eq:diff_eq_system_u} remains. Moreover, since $\dot{\mathbf{x}}^{*}$ does not depend on $\mathbf{u}$, the Jacobians in Theorems~\ref{Theorem:ISS}-\ref{Theorem:u_dynamics} do not change and the statements in Theorems~\ref{Theorem:ISS}-\ref{Theorem:u_dynamics} hold.
\end{proof}

The importance of Corollary~\ref{corollary} is that our herding solution can be seen as a generalization of previous herding works~\cite{Pierson_2018_TR_Herding}~\cite{song2021herding} where the purpose is to drive the evaders to a certain region. By imposing a time-varying reference, herders not only drive evaders towards particular positions but also steer them towards a particular region whereas maintaining a specific formation. Besides, note that no particular restrictions regarding the initial configuration of the entities is specified. Nevertheless, it is convenient that the herders surround the evaders to avoid escapes before the beginning of the herding, so in Subsection~\ref{subsec:approaching} we develop a secure caging stage.

From an algorithmic point of view, the calculation of the control input is very simple. At each instant, the controller receives $\mathbf{x}$ and $\mathbf{u}$ from an observer and/or from measurements. Then, $f^*(\mathbf{x})$ is computed, which, together with the dynamic model of the evaders $f(\mathbf{x},\mathbf{u})$ (e.g.,~\eqref{eq:PiersonBase} or~\eqref{eq:LicitraBase}), allows to compute $h^*(\mathbf{x},\mathbf{u})$ with Eq.~\eqref{eq:h_star}. Besides, the Jacobians $\mathbf{J}_{\mathbf{x}}$ and $\mathbf{J}_{\mathbf{u}}$ can be calculated either analytically or numerically, depending on their complexity. Finally, $\dot{\mathbf{u}}$ is computed from Eq.~\eqref{eq:diff_eq_system_u}. Regarding the estimation of $\mathbf{x}$ and $\mathbf{u}$, another advantage of Implicit Control is that $\mathbf{u}$ now belongs to the state of an expanded system. Therefore, it is possible to apply the same estimation techniques as with $\mathbf{x}$, robustifying the herding solution.


\section{Adding adaptation}\label{sec:adaptation}

In the previous Section perfect knowledge of the evaders' dynamics is assumed. However, in real problems there exist sources of uncertainty. We derive a general adaptation law which overcomes this issue and preserves the control properties.

Consider the parameters $\theta_j$ in Eqs.~\eqref{eq:PiersonBase} and~\eqref{eq:LicitraBase}. If we define $\theta = \hbox{diag}(\mathbf{\theta}_1,\hdots,\mathbf{\theta}_m) \in \mathbb{R}^{m \times m}$, then the original system~\eqref{eq:initial_system} is rewritten to the following, with a little abuse of notation for the sake of readability,
\begin{equation}\label{eq:original_theta}
    \dot{\mathbf{x}} = f_{\theta}(\mathbf{x}, \mathbf{u},\mathbf{\theta}) = \mathbf{\theta} f(\mathbf{x},\mathbf{u}).
\end{equation}
We assume a linear dependency with respect to the uncertain parameters, and that the parameters are constant. This assumption is typical in adaptive control theory and, in the case of Implicit Control, implies that the gradients with respect to the parameters do not depend on the parameters, which is the key that yields to the proposed adaptation law. Nevertheless, despite simple, linear parameters can still model the aggressiveness in the repulsion provoked by the herders. Besides, notice that $\theta_j > 0$ because otherwise the evaders do not move or are not repelled by the herders.

We represent the uncertainty in $\mathbf{\theta}$ by the vector of estimates $\hat{\mathbf{\theta}}$, defining the error between the real and estimated parameters as $\tilde{\mathbf{\theta}} = \mathbf{\theta} - \hat{\mathbf{\theta}}$. From now on, the symbols $\hat{\cdot}$ and $\tilde{\cdot}$ denote the quantities that depend on $\hat{\theta}$ and $\tilde{\theta}$. In general, $\tilde{\mathbf{\theta}} \neq \mathbf{0}$ so Implicit Control no longer guarantees stability. The difference between actual and estimated dynamics can be defined by a new function
\begin{equation}\label{eq:def_PSI}
    \tilde{h}(\mathbf{x}, \mathbf{u},\tilde{\theta}) = h(\mathbf{x}, \mathbf{u},\mathbf{\theta}) - \hat{h}(\mathbf{x}, \mathbf{u},\hat{\mathbf{\theta}}) = \tilde{\mathbf{\theta}}f(\mathbf{x}, \mathbf{u}).
\end{equation}
This is a standard adaptive control problem, for which there exist different solutions for both linear and nonlinear systems~\cite{Khalil2014Nonlinear}. Nevertheless, we propose a different adaptation law, following the principles of Implicit Control, that simplifies the stability analysis.

The objective of the adaptation law is to design $\dot{\hat{\mathbf{\theta}}}$ such that $\tilde{h}$ converges to zero while the closed-loop behavior of the system is preserved. The structure of the problem, and, in particular, of $\tilde{h}$, is similar to that of Implicit Control. We can follow the same input design procedure of Implicit Control, in this case over $\tilde{h}$, to design an adaptation law of the form
\begin{equation}\label{eq:adaptive_law}
    \dot{\hat{\mathbf{\theta}}} = -\tilde{\mathbf{J}}_{\mathbf{\theta}}^{-1}(\tilde{h}^* + \hat{\mathbf{J}}_{\mathbf{x}} \dot{\hat{\mathbf{x}}} +  \hat{\mathbf{J}}_{\mathbf{u}} \dot{{\mathbf{u}}}).
\end{equation}
Here,  
\begin{equation}\label{eq:J_thetaxthetau}
\begin{aligned} 
    \kern -0.3cm  \hat{\mathbf{J}}_{\mathbf{x}}  \kern -0.1cm = \kern -0.1cm
    \left.\begin{pmatrix}
    \frac{\partial \hat{h}_{1}}{\partial {\mathbf{x}}_1} & \kern -0.3cm \hdots \kern -0.3cm & \frac{\partial \hat{h}_{1}}{\partial {\mathbf{x}}_m}
    \\
    \vdots & \kern -0.3cm \ddots \kern -0.3cm & \vdots 
    \\
    \frac{\partial \hat{h}_{m}}{\partial {\mathbf{x}}_1} & \kern -0.3cm \hdots \kern -0.3cm & \frac{\partial \hat{h}_{m}}{\partial {\mathbf{x}}_m}
    \end{pmatrix}\right.,
& \hbox{ }
    \hat{\mathbf{J}}_{\mathbf{u}}  \kern -0.1cm = \kern -0.1cm
    \left.\begin{pmatrix}
    \frac{\partial \hat{h}_{1}}{\partial {\mathbf{u}}_1} & \kern -0.3cm \hdots \kern -0.3cm & \frac{\partial \hat{h}_{1}}{\partial {\mathbf{u}}_n}
    \\
    \vdots & \kern -0.3cm \ddots  \kern -0.3cm & \vdots 
    \\
    \frac{\partial \hat{h}_{m}}{\partial {\mathbf{u}}_1} & \kern -0.3cm \hdots \kern -0.3cm & \frac{\partial \hat{h}_{m}}{\partial {\mathbf{u}}_n}
    \end{pmatrix}\right.,
\end{aligned}
\end{equation}
$\tilde{\mathbf{J}}_{\mathbf{\theta}} = \hbox{diag}(f)$, and $\tilde{h}^*$ is a free-design function. 

As in the design of the input dynamics, $\tilde{h}^*$ determines the desired dynamic behavior for the derivative of~\eqref{eq:def_PSI}. Concatenating the new set of equations yields a new expanded system
\begin{equation}\label{eq:diff_eq_adapt}
\left\{
\begin{aligned}
    \dot{\mathbf{x}} &= \theta f
    \\ 
    \dot{{\mathbf{u}}} &= \hat{\mathbf{J}}_{\mathbf{u}}^{+}
    \kern -3pt\left(\hat{h}^*-\hat{\mathbf{J}}_{\mathbf{x}} \dot{\hat{\mathbf{x}}} \right)
    \\ 
    \dot{\hat{\mathbf{\theta}}} &= -\tilde{\mathbf{J}}_{\mathbf{\theta}}^{-1}(\tilde{h}^* + \hat{\mathbf{J}}_{\mathbf{x}} \dot{\hat{\mathbf{x}}} +  \hat{\mathbf{J}}_{\mathbf{u}} \dot{{\mathbf{u}}})
\end{aligned}
\right.. 
\end{equation}
Acknowledging the similarities between~\eqref{eq:diff_eq_adapt} and~\eqref{eq:diff_eq_system_u}, we extend the results in Theorems~\ref{Theorem:ISS} and \ref{Theorem:u_dynamics}.
\begin{proposition}
\label{proposition:ISS}
If the system defined in~\eqref{eq:initial_system} is ISS and Assumption~\ref{assumption:conditions} holds, then the adaptive law in~\eqref{eq:adaptive_law} preserves the properties of the Implicit Control if
\begin{equation}\label{eq:right_hstar}
    \tilde{h}^* = h^*_{\theta} - \dot{h} + \hat{\mathbf{J}}_{\mathbf{x}}(\dot{\mathbf{x}}-\dot{\hat{\mathbf{x}}}),
\end{equation}
where $h^*_{\theta}$ is a free-design function chosen to make $\tilde{h}$ stable.
\end{proposition}
\begin{proof}
The proof follows some of the steps in Theorem~\ref{Theorem:ISS}. Let consider the chain rule expansion of $d\tilde{h}/dt$,
\begin{equation}\label{eq:der_h_theta}
    \frac{d \tilde{h}}{d t} = \frac{\partial \tilde{h}}{\partial \mathbf{x}}\dot{\mathbf{x}} + \frac{\partial \tilde{h}}{\partial \mathbf{u}}\dot{\mathbf{u}} + \frac{\partial \tilde{h}}{\partial \tilde{\theta}}\dot{\tilde{\theta}}.
\end{equation}
Since $\theta$ is constant, $\dot{\tilde{\theta}} = -\dot{\hat{\theta}}$ and the substitution of Eq.~\eqref{eq:adaptive_law} gives
\begin{equation}\label{eq:der_h_theta2}
    \frac{d \tilde{h}}{d t} = {\mathbf{J}}_{\mathbf{x}}\dot{\mathbf{x}} - \hat{\mathbf{J}}_{\mathbf{x}}\dot{\mathbf{x}} + {\mathbf{J}}_{\mathbf{u}}\dot{\mathbf{u}} - \hat{\mathbf{J}}_{\mathbf{u}}\dot{\mathbf{u}}  + \tilde{h}^* + \hat{\mathbf{J}}_{\mathbf{x}} \dot{\hat{\mathbf{x}}} +  \hat{\mathbf{J}}_{\mathbf{u}} \dot{{\mathbf{u}}}.
\end{equation}
Note that $ {\mathbf{J}}_{\mathbf{x}}\dot{\mathbf{x}} + {\mathbf{J}}_{\mathbf{u}}\dot{\mathbf{u}}$ is equal to $\dot{h}$, and the $\hat{\mathbf{J}}_{\mathbf{u}} \dot{{\mathbf{u}}}$ terms cancel out. Therefore, Eq.~\eqref{eq:der_h_theta2} can be written as
\begin{equation}\label{eq:der_h_theta4}
    \frac{d \tilde{h}}{d t} = \dot{h} + \tilde{h}^* - \hat{\mathbf{J}}_{\mathbf{x}}(\dot{\mathbf{x}}-\dot{\hat{\mathbf{x}}}) = h^*_{\theta},
\end{equation}
where the later follows from the definition of $\tilde{h}^*$ in Eq.~\eqref{eq:right_hstar}.
Now, let consider the chain rule expansion of $d\hat{h}/dt$,
\begin{equation}\label{eq:der_h_theta_hat}
    \frac{d \hat{h}}{d t} = \frac{\partial \hat{h}}{\partial \mathbf{x}}\dot{\mathbf{x}} + \frac{\partial \hat{h}}{\partial \mathbf{u}}\dot{\mathbf{u}} + \frac{\partial \hat{h}}{\partial \hat{\theta}}\dot{\hat{\theta}}.
\end{equation}
If we use the definitions of $\dot{\mathbf{u}}$ and $\dot{\hat{\theta}}$ in Eq.~\eqref{eq:diff_eq_adapt}, Eq.~\eqref{eq:der_h_theta_hat} reads
\begin{equation}\label{eq:der_h_theta_hat2}
    \frac{d \hat{h}}{d t} = \hat{\mathbf{J}}_{\mathbf{x}}\dot{\mathbf{x}} + \hat{h}^* - \hat{\mathbf{J}}_{\mathbf{x}} \dot{\hat{\mathbf{x}}} - \tilde{h}^* - \hat{\mathbf{J}}_{\mathbf{x}} \dot{\hat{\mathbf{x}}} - \hat{h}^* + \hat{\mathbf{J}}_{\mathbf{x}} \dot{\hat{\mathbf{x}}}.
\end{equation}
By exploiting $\tilde{h}^*$ in Eq.~\eqref{eq:right_hstar}, we get
\begin{equation}\label{eq:der_h_theta_hat3}
    \frac{d \hat{h}}{d t} = \dot{h} -  h^*_{\theta}.
\end{equation}
Finally, let consider the chain rule expansion of $d h/dt$
\begin{equation}\label{eq:der_h_theta_h}
    \frac{d {h}}{d t} = \frac{\partial {h}}{\partial \mathbf{x}}\dot{\mathbf{x}} + \frac{\partial {h}}{\partial \mathbf{u}}\dot{\mathbf{u}} + \frac{\partial {h}}{\partial {\theta}}\dot{{\theta}} = {\mathbf{J}}_{\mathbf{x}}\dot{\mathbf{x}} + {\mathbf{J}}_{\mathbf{u}}\dot{\mathbf{u}}.
\end{equation}
The substitution of $\dot{\mathbf{u}}$ in Eq.~\eqref{eq:diff_eq_adapt} gives
\begin{equation}\label{eq:der_h_theta_h2}
    \frac{d {h}}{d t} = {\mathbf{J}}_{\mathbf{x}}\dot{\mathbf{x}} + {\mathbf{J}}_{\mathbf{u}}\hat{\mathbf{J}}_{\mathbf{u}}^{+}
    \kern -3pt\left(\hat{h}^*-\hat{\mathbf{J}}_{\mathbf{x}} \dot{\hat{\mathbf{x}}} \right) = {\mathbf{J}}_{\mathbf{x}} (\dot{\mathbf{x}} - \dot{\hat{\mathbf{x}}}) + h^* = {\mathbf{J}}_{\mathbf{x}} \dot{\tilde{h}} + h^*.
\end{equation}
With Eqs.~\eqref{eq:der_h_theta4},~\eqref{eq:der_h_theta_hat3} and~\eqref{eq:der_h_theta_h2} we use the same arguments in the proof of Theorem~\ref{Theorem:ISS}. Let $\tilde{h} \neq \mathbf{0}$ be considered as a perturbation. Due to Eq.~\eqref{eq:der_h_theta4}, $\tilde{h}$ converges to zero at some point. In the meantime, since the system is ISS, the stability remains. Besides, $\tilde{h}$ is stable, so the perturbation vanishes with time. Then, according to Eq.~\eqref{eq:der_h_theta_hat3}, $\dot{\hat{h}} = \dot{h}$ and, according to Eq.~\eqref{eq:der_h_theta_h2}, $\dot{h} = h^*$. Therefore, we recover the expanded system in Eq.~\eqref{eq:diff_eq_system_u} and the results of Theorem~\ref{Theorem:ISS} hold.
\end{proof}
\begin{proposition}
\label{proposition:u_dynamics}
Given Assumption~\ref{assumption:conditions}, 
\begin{equation}\label{eq:h_star_theta}
    h_{\theta}^*(\mathbf{x},\mathbf{u}) =
    - \mathbf{K}_{\theta} (\mathbf{x},\mathbf{u})\tilde{h}
\end{equation}
with $\mathbf{K}_{\theta}(\mathbf{x},\mathbf{u})$ a function which returns a positive definite matrix, and $\mathbf{K}(h(\mathbf{x},\mathbf{u}))$ defined as in~\eqref{eq:negative_matrix}, if 
\begin{equation}\label{eq:negative_matrix_adap}
    \kern -0.2cm \bar{\mathbf{K}} \kern -0.1cm = \kern -0.1cm \left.\begin{pmatrix}
    -\mathbf{K}_f(\mathbf{x}) \kern -0.1cm & \kern -0.1cm 0.5\mathbf{I} \kern -0.1cm & \kern -0.1cm \mathbf{0}
    \\
    0.5\mathbf{I} \kern -0.1cm & \kern -0.1cm -\mathbf{K}_h(h(\mathbf{x},\mathbf{u})) \kern -0.1cm & \kern -0.1cm -0.5{\mathbf{J}}_{\mathbf{x}}\mathbf{K}_{\theta}(\mathbf{x},\mathbf{u})
    \\
    \mathbf{0} \kern -0.1cm & \kern -0.1cm -0.5{\mathbf{J}}_{\mathbf{x}}\mathbf{K}_{\theta}(\mathbf{x},\mathbf{u}) \kern -0.1cm & \kern -0.1cm - \mathbf{K}_{\theta}(\mathbf{x},\mathbf{u}) 
    \end{pmatrix}\right.
\end{equation}
is negative definite, then the adaptive law preserves the stability properties of Theorem~\ref{Theorem:u_dynamics}.
\end{proposition}
\begin{proof}
The proof follows from Theorem~\ref{Theorem:u_dynamics} by considering 
\begin{equation}\label{eq:Ly_theorem_propo}
    V = \frac{1}{2}\mathbf{x}^T\mathbf{x} + \frac{1}{2}{h}^T{h} + \frac{1}{2}\tilde{h}^T \tilde{h}
\end{equation}
as a Lyapunov candidate function. Developing its derivative and using Eqs.~\eqref{eq:der_h_theta4} and~\eqref{eq:der_h_theta_h2} from Proposition~\ref{proposition:ISS}, we obtain
\begin{equation}\label{eq:Ly_theorem_der_3_propo1}
\begin{aligned}
    \dot{V} =& \mathbf{x}^T h + \mathbf{x}^T f^* + h^T \dot{h} + \tilde{h}^T \dot{\tilde{h}} =
    \\ &
    \mathbf{x}^T h - \mathbf{x}^T \mathbf{K}_f \mathbf{x} - h^T {\mathbf{J}}_{\mathbf{x}}\mathbf{K}_{\theta} \tilde{h}- h^T \mathbf{K}_h h - \tilde{h}^T \mathbf{K}_{\theta} \tilde{h}.
\end{aligned}
\end{equation}
As in Theorem~\ref{Theorem:u_dynamics}, the latter reads
\begin{equation}\label{eq:Ly_theorem_der_3_propo}
    \dot{V} = \begin{pmatrix} \mathbf{x}^T & {h}^T & \tilde{h}^T \end{pmatrix} \bar{\mathbf{K}} \begin{pmatrix} \mathbf{x}^T & {h}^T & \tilde{h}^T \end{pmatrix}^T.
\end{equation}

Then, if $\mathbf{K},\mathbf{K}_{\theta}$ are designed such that $\bar{\mathbf{K}}$ is negative definite, then the adaptation law preserves the stability properties of Theorem~\ref{Theorem:u_dynamics}.
\end{proof}

The adaptation law, then, evolves the estimates towards the values which ensure the correct behavior of the control. Remarkably, the convergence of $\tilde{h}$ to zero implies that $\hat{h}$ converges to $h$. Therefore, the convergence of $\hat{h}$ to zero implies the convergence of $h$ to zero, and the evaders behave as desired. 
In practice, to implement the controller, a slight change in Eq.~\eqref{eq:diff_eq_adapt} is required to compute $\tilde{h}^*$. Since $\dot{\mathbf{x}}=\theta f$ is not available, an approximation (e.g., by finite differences) must be used to obtain $\dot{\mathbf{x}}$. Note that this only affects the implementation, as long as the stability analysis involves the real dynamics of the evaders.

\textbf{Example.} We provide a brief description of how to design the control and adaptation gains. 

The first step is to set the desired behavior of the herd. Theorem~\ref{Theorem:u_dynamics} and Proposition~\ref{proposition:u_dynamics} allow for several possible choices, but to keep it simple, consider a desired settling time of $1/\tau = 12$s and an exponential transient according to $m$ first order independent systems. Thus,
\begin{equation}
    f^*(\mathbf{x})=-\mathbf{K}_f(\mathbf{x})\mathbf{x}=-\mathbf{K}_f\mathbf{x} = -(3\tau\mathbf{I})\mathbf{x} = -0.25\mathbf{I}\mathbf{x}.
\end{equation}
Then, according to Theorem~\ref{Theorem:u_dynamics}, $\mathbf{K}_h (h(\mathbf{x},\mathbf{u}))$ in Eq.~\eqref{eq:negative_matrix} must be such that $\mathbf{K}$ is negative definite. To preserve simplicity, let $\mathbf{K}_h (h(\mathbf{x},\mathbf{u})) = \mathbf{K}_h$ be a constant matrix. Under these conditions, for example $\mathbf{K}_h = 50\mathbf{I}$ works.

In the same way, if an adaptation law is desired, then $\mathbf{K}_{\theta} (\mathbf{x},\mathbf{u})$ in Eq.~\eqref{eq:negative_matrix_adap} must be such that $\bar{\mathbf{K}}$ is negative definite. In the herding case, the speed of evaders and herders is bounded, so $\mathbf{J}_{\mathbf{x}}$ is bounded and there always exist a matrix $\mathbf{K}_{\theta} (\mathbf{x},\mathbf{u})$ that makes $\bar{\mathbf{K}}$ negative definite. To preserve simplicity again, let $\mathbf{K}_{\theta} (\mathbf{x},\mathbf{u}) = \mathbf{K}_{\theta}$ be a constant matrix. For instance, $\mathbf{K}_{\theta} = 200\mathbf{I}$ works in the examples of Section~\ref{sec:simulations}. 
In general, starting from $\mathbf{K}_f,$ choosing the other gains such that $||\mathbf{K}_f||<||\mathbf{K}_h||<||\mathbf{K}_{\theta}||$ will make the system stable.
Nevertheless, there are no particular restrictions on the gains, providing a lot of design possibilities. If needed/desired, $\mathbf{K}_f$ might not be a constant, nor $f^*$ a linear function of the state. Similarly, $\mathbf{K}_h$ can change with time to preserve the negative definiteness of $\mathbf{K}$, and the same applies to $\mathbf{K}_{\theta}$.


\section{Herding solution}\label{sec:distributed}

Implicit Control, along with the adaptation law, constitutes the core of our proposal.
Nonetheless, to have a complete herding solution, additional issues are considered.

\subsection{Distributed estimator}\label{subsec:dekf}

Up to now, the solution assumes that herders know the state and input of the system perfectly, i.e., the position of all the herders and evaders. 
This Subsection describes a distributed estimator that fulfills this assumption.
Particularly,
we develop an extended version of the Distributed Kalman Filter~\cite{olfati2007distributed} (E-DKF), tailored to work with the Implicit Control expanded dynamics. The objective is, for all the herders, to have an accurate estimation of evaders' and herders' position, aggregated in the vector $\xi = \begin{bmatrix}
    \mathbf{x}_{1}^{T} & \dots  & \mathbf{x}_{m}^{T} &
    \mathbf{u}_{1}^{T} & \dots  & \mathbf{u}_{n}^{T}
    \end{bmatrix}^{T}.$
Notice that, in contrast with typical DKFs, the control inputs must be included as quantities to be estimated.
In typical control solutions this is not possible because the input has no dynamics because it is just an algebraic expression. 

Implicit Control addresses this issue by providing the input dynamics. In particular, for the prediction step of the filter we consider
\begin{equation}\label{eq:system}
    \dot{\xi} = \left.\frac{\partial \Sigma(\xi)}{\partial \xi} \right|_{\xi=\xi(t)} + \mathbf{w},
\end{equation}
where $\Sigma(\xi)$ denotes the expanded system in~\eqref{eq:diff_eq_system_u} and $\mathbf{w} \sim \mathcal{N}(\mathbf{0},\mathbf{Q})$ is a zero mean Gaussian noise with covariance $\mathbf{Q}$. Note that the linearization in Eq.~\eqref{eq:system} is necessary because the overall expanded dynamics are strongly nonlinear, so we use $\frac{\partial \Sigma(\xi)}{\partial \xi}$ evaluated at $\xi=\xi(t)$. Since the evaders' dynamics are deterministic, $\mathbf{Q}$ will have small values, useful to mitigate the effects of the uncertainty during the adaptation law's transient.

For the filter update, it is important to note that the herders' measurements will depend on its position, i.e., a herder will only be able to measure the entities that are close to it. Therefore, the measurement equation is of the form
\begin{equation}\label{eq:distributedSensorModel}
    \mathbf{z}_i = \mathbf{H}_i(\xi)\xi + \mathbf{v}_i,
\end{equation}
where $\mathbf{z}_i$ is the measurement taken by herder $i$, $\mathbf{v}_i \sim \mathcal{N}(\mathbf{0},\mathbf{R}_i)$ is the measurement noise with covariance $\mathbf{R}_i$, and $\mathbf{H}_i$ depends on $\xi$. This dependency makes the dimension of $\mathbf{z}_i$ to change at each filter iteration, complicating the implementation. 

The solution lies in using an information form of the E-DKF so, at each instant, each herder has a vector of measurements and a vector of booleans. The latter says whether an entity has been sensed or not, whereas the former contains the actual measurements and zeroes for the rest of entities, i.e., $\mathbf{z}_i \in \mathbb{R}^{2m+2n}$.
Given the two vectors, the information form of the E-DKF permits to add up all the measurements in common state coordinates, only adding those measurements where the boolean vector is equal to one. Besides, communication among herders happens only if the distance between two herders is less than $d_c$. Thus, the estimator is \textit{input-dependent} in the topology as well.

It is noteworthy that the stochastic nature of the estimator can absorb discrepancies from the adaptation and the unmodeled perturbations. This, together with the information form of the filter, permits to run the control and the estimation at different frequencies. The control loop can be the fastest while the estimation runs slower to account for communication bandwidth issues. Regarding the influence of the estimator in the overall stability of the system, if the estimator is designed to converge sufficiently fast, from a practical point of view its influence in the controller can be overlooked, and the controller works with accurate states and inputs. Regarding global awareness, the estimator still assumes a common reference frame for all the agents (herders and evaders) as well as knowledge of their number. Nevertheless, there are works available that deal with the distributed computation of these quantities~\cite{Franceschelli_TRO_2013_CommonReference,Montijano_TCST_2014_DynamicOpinions}.

\subsection{Caging the evaders}\label{subsec:approaching}

To avoid escapes, herders must be close to the evaders before beginning the precision herding using Implicit Control. This is not a complete solution because robots might not be deployed in the proximity of the evaders. Therefore, we propose a caging stage to prevent these escapes. 

We consider a caging strategy based on the mean position and covariance of the herd. We denote them by $\mathbf{x}_{avg} := \frac{1}{m} \sum_j \mathbf{x}_j$ and $\mathbf{P}_{avg} := \frac{1}{m} \sum_j(\mathbf{x}_j-\mathbf{x}_{avg})(\mathbf{x}_j-\mathbf{x}_{avg})^T$. The combination of both gives a description of the herd in terms of an ellipsoid centered in $\mathbf{x}_{avg}$. The geometry of the setup for a 2D herding is depicted in Fig.~\ref{fig:approaching}. It can be observed that the ellipse (in red) delimits a region around the evaders (in green) which suggests a uniform distribution of the herders along the perimeter. However, these locations might be too close to the evaders. Therefore, it is convenient to use an augmented version of the ellipse, $\mu_1 \mathbf{P}_{avg}$ with $\mu_1 > 1$, to place the herders sufficiently far from the evaders. We do not make any assumptions about the initial distribution of the evaders. For instance, a very spread distribution of evaders with a large $\mathbf{P}_{avg}$ will only imply that it will take more time for the herders to surround the evaders. Moreover, $\mu_1$ can be tuned so that herders provoke weak repulsive forces in the evaders. During the caging, evaders will not escape, and afterwards, the herders will approach the evaders to provoke large enough repulsive forces to steer the herd. Anyway, this is transparent for both the caging and the precise herding stage. 

To solve the caging, we propose a controller, adapted from~\cite{Monti_CDC_2013_Entrapment}, that steers the herders towards an even distribution along the perimeter of $\mu_1 \mathbf{P}_{avg}$. The position of herder $i$, in polar coordinates with respect to the ellipse, is driven by 
\begin{align}
    \dot{\psi}_i &= (\psi_{i+1} - \psi_{i}) + (\psi_{i-1} - \psi_{i}) \label{eq:nominal_controller_psi}
    \\
    \dot{\rho}_i &= \rho_i - \rho_i^* \label{eq:nominal_controller_rho}
\end{align}
where $\psi_i$ and $\rho_i$ are the angular and radial positions, and $i+1$, $i-1$ denote the left and right closest herders. The policy is inspired in the well known $N$-bugs problem, that evenly distributes $N$ agents over a closed curve.
An even distribution of herders along the perimeter of an ellipse requires the herders to be distributed with a uniform angle. The agreement in the angle is given by Eq.~\eqref{eq:nominal_controller_psi}. Meanwhile, to each angle corresponds a certain radial position to place herders over the desired ellipse, given by $\rho_i^*$. Therefore, the controller of the radial position is just the linear feedback controller in Eq.~\eqref{eq:nominal_controller_rho}. For formal guarantees of convergence, we refer to~\cite{Monti_CDC_2013_Entrapment}. 

Herder $i$ only needs to measure the closest left and right neighbor. These commands are then translated to Cartesian coordinates. Under this control policy the herders arrive to equally-distributed positions along the perimeter of the ellipse, denoted by $\mathbf{u}^*_i$ and depicted in blue in Fig.~\ref{fig:approaching}. 

\vspace{-0.4cm}

\begin{figure}[!ht]
    \centering
        \includegraphics[width=0.7\columnwidth]{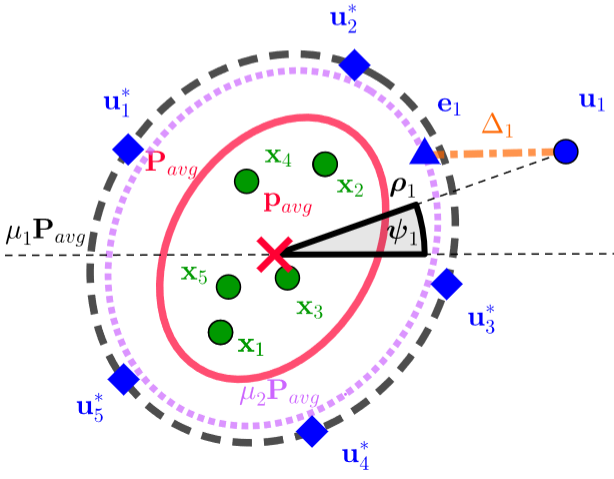}
	\caption{Geometrical view of the caging stage. The evaders (green dots) are in the region determined by their mean position and covariance (in red). To avoid escapes and high repulsive forces, the herder (blue dot) must be placed in the perimeter (blue diamonds) of an augmented ellipse (black-dashed). Besides, since there is uncertainty, CBFs prevent the herders to enter in the magenta ellipse, computing the closest point to its perimeter (blue triangle). }
	\label{fig:approaching}
\end{figure}

Due to perturbations or unmodeled phenomena, the caging might not prevent herders to move too close to the evaders. To solve this issue, we propose the use of CBFs~\cite{wang2017safety}. The controller in Eqs.~\eqref{eq:nominal_controller_psi}-\eqref{eq:nominal_controller_rho} computes the velocity $\dot{\mathbf{u}}_i^{nom}$, which is then translated into Cartesian coordinates. This is called the \textit{nominal action}. The nominal action might move the herder towards a region where the interaction with the evaders is not convenient. We define this region by $\mu_2 \mathbf{P}_{avg}$, where $\mu_1 > \mu_2 > 1$, and it is depicted in purple in Fig.~\ref{fig:approaching}.

To correct the nominal action, we use a quadratic program based on CBFs. In particular, at each instant, the closest point in the perimeter of the ellipse $\mu_2 \mathbf{P}_{avg}$ to herder $i$ is computed, denoted by $\mathbf{e}_i$~\cite{eberly2011distance}.
Then, an optimization process~\cite{wang2017safety} is carried out,
\begin{subequations}\label{eq:CBF}
\begin{alignat}{2}
&\:\:\:\:\:\:\:\dot{\mathbf{u}}_i^{cbf} = \underset{\dot{\mathbf{u}}_i^{cbf}}{\arg \min}  \:\: ||\dot{\mathbf{u}}_i^{cbf} - \dot{\mathbf{u}}_i^{nom}||^2  \label{eq:CBFobj}
\\
s.t. & \:\:\:\:\:\:\: L_f h(\Delta_i) + L_g h(\Delta_i) \dot{\mathbf{u}}_i^{cbf} \geq -\kappa(h(\Delta_i)) \label{eq:CBFconstraint}
\end{alignat}
\end{subequations}
Here, $\Delta_i = \mathbf{u}_i - \mathbf{e}_i$, depicted in orange in Fig.~\ref{fig:approaching}. Moreover, $\dot{\mathbf{u}}_i^{cbf}$ is the action to be optimized, $h(\Delta_i)$ is the control barrier function, $\kappa(h(\Delta_i))$ is a class $\mathcal{K}$ function and $L_f$, $L_g$ denote the Lie derivatives. It is important to remark that the abuse of notation with function $h$ is to fit the typical notation of CBFs~\cite{wang2017safety}.

Finally, a high-level algorithm of the complete herding sequence is described in Algorithm~\ref{al:with_Ad}. 

\begin{algorithm}
\caption{Complete herding for herder $i$}\label{al:with_Ad}
\begin{algorithmic}[1]
    \STATE Receive $\mathbf{x}_{avg}$ and $\mathbf{P}_{avg}$
    \WHILE{True}
        \STATE Compute estimate $\bar{\xi}_i$ with E-DKF
        \IF{ herders are far from evaders }
            \STATE \textbf{Do caging}: compute $\dot{\psi}_i$, $\dot{\rho}_i$, calculate $\dot{\mathbf{u}}_i^{cbf}$
        \ELSE 
            \STATE \textbf{Do precise herding}: compute $\dot{\mathbf{u}}_i$, $\dot{\hat{\theta}}_i$ with Adaptive Implicit Control, and update $\hat{\theta}_i$
        \ENDIF
        \STATE Update $\mathbf{u}_i$
    \ENDWHILE
\end{algorithmic}
\end{algorithm}


\section{Simulation results}\label{sec:simulations}

The purpose of these simulations is to validate and demonstrate the success of the proposal against challenging situations using the models in Section~\ref{sec:prosta}. The herding solution is formed by different parts, so the simulations will be conducted progressively to analyze each individual component. 

\subsection{Implicit Control and adaptation law}

The first case of study consists in the herding of 5 evaders by 5 herders. The herders receive perfect measurements and they are already deployed near the evaders. The details are in Table~\ref{table:initial_values}. We set $\mathbf{K}_f = 0.25 \mathbf{I}_{2 m}$, yielding a settling time of $12$s with an exponential transient according to $2m$ first order independent systems. To ensure that the conditions of Theorem~\ref{Theorem:u_dynamics} and Proposition~\ref{proposition:u_dynamics} hold, we set $\mathbf{K}_h = 50 \mathbf{I}_{2 m}$ and $\mathbf{K}_{\theta} = 200 \mathbf{I}_{2 m}$. For the sake of comparison, we use the \textit{Levenberg-Marquardt} numerical method~\cite{Marquardt1963LM} as a \textit{baseline} to compute the control input.
\vspace{-0.2cm}
\begin{table}[!ht]
\renewcommand{\arraystretch}{1.64}
\centering
\caption{List of symbols in the figures.}
\begin{tabular}{|c|c|}
    \hline
    Symbol
    &
    Meaning
    \\
    \hline
    \raisebox{-0.5\totalheight}{\includegraphics[width=0.10\columnwidth]{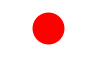}}
    &
    \raisebox{-0.2\totalheight}{{\footnotesize Desired positions}}
    \\
    \hline
    \begin{tabular}{c|c}
    Symbol
    &
    Meaning
    \\
    \hline
    \raisebox{-.5\totalheight}{\includegraphics[width=0.10\columnwidth]{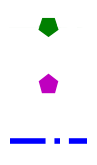}}
    &
    \begin{tabular}{c}
    {\footnotesize Initial position herders}
    \\
    {\footnotesize Final position herders}
    \\
    {\footnotesize Trajectories herders}
    \end{tabular}
    \end{tabular}
    &
    \begin{tabular}{c|c}
    Symbol
    &
    Meaning
    \\
    \hline
    \raisebox{-.5\totalheight}{\includegraphics[width=0.10\columnwidth]{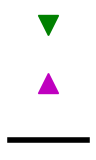}}
    &
    \begin{tabular}{c}
    {\footnotesize Initial position evaders}
    \\
    {\footnotesize Final position evaders}
    \\
    {\footnotesize Trajectories evaders}
    \end{tabular}
    \end{tabular}
    \\
    \hline
\end{tabular}
\label{table:legend}
\end{table}
\vspace{-0.4cm}
\begin{table}[!ht]
\centering
\caption{Parameters of the simulations.}
\begin{tabular}{c}
\begin{tabular}{|c|c|c|c|c|c|c|c|c|}
\hline
 $\theta_j$ (Inv.) &$\theta_j$ (Exp.)&$\beta_j$&$\sigma_j$& $d_{min}$ & $T$ & $v_{max}$ & $d_m$ & $d_c$ \\
\hline
$1.0$ &$0.5$ &$0.5$ & $2.0$ & $1.0$m & $10$ms & $0.4$m/s & $6.5$m & $6.5$m\\
\hline
\end{tabular}
\end{tabular}
\label{table:initial_values}
\end{table}

The first row of Fig.~\ref{fig:sim_results} shows the trajectories followed by herders and evaders for the different test cases. Implicit Control is able to herd the evaders successfully, achieving almost the same trajectories obtained by the numerical baseline in all the experiments. Conversely, the trajectories of the herders present some differences depending on the motion model, which highlights the complexity of the control problem at hand, greater for the Exponential Model due to the switching dynamics. The second row of Fig.~\ref{fig:sim_results} represents the evolution over time of the evaders' position error. Implicit Control achieves the desired settling time of $12$s and the state evolves as an exponential function according to the imposed closed-loop dynamics. Interestingly, it is corroborated that the Implicit Control solution does not assign any particular evader to any robotic herder, so herders just move to steer all the evaders towards their particular assigned positions \textit{simultaneously}.
\textit{All} herders are contributing at \textit{all} instants to the motion of \textit{all} evaders. But, to achieve a stable equilibrium, a team of herders is needed because otherwise the herding can not be simultaneous, as it is seen in the second row of Fig.~\ref{fig:sim_results}. At equilibrium, each herder is compensating the repulsive forces provoked by the other herders, reason why the evaders end up in the desired positions with zero velocity.

\begin{figure*}[!ht]
    \centering
    \begin{tabular}{cccc}
         \hspace{0.4cm}{\footnotesize Inverse, Implicit Control}
         &  
         \hspace{0.4cm}{\footnotesize Exponential, Implicit Control}
         &
         \hspace{0.4cm}{\footnotesize Inverse, numerical baseline}
         &
         \hspace{0.4cm}{\footnotesize Exponential, numerical baseline}
         \\\hspace{-0.35cm}\hspace{0.1cm}
         \includegraphics[width=0.21\textwidth]{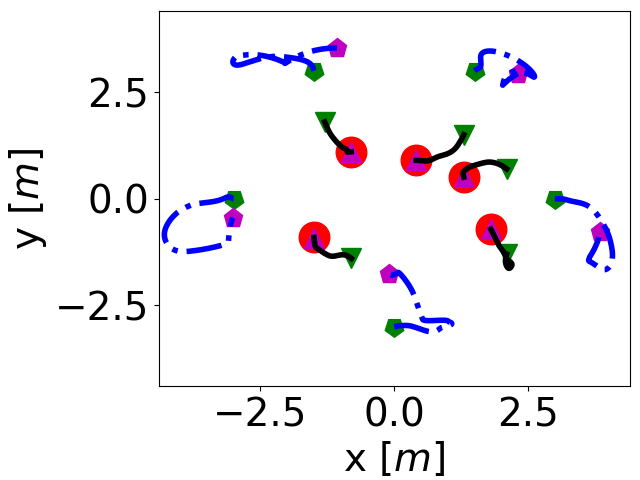}\hspace{0.1cm}
         &\hspace{-0.1cm}\hspace{0.1cm}
         \includegraphics[width=0.21\textwidth]{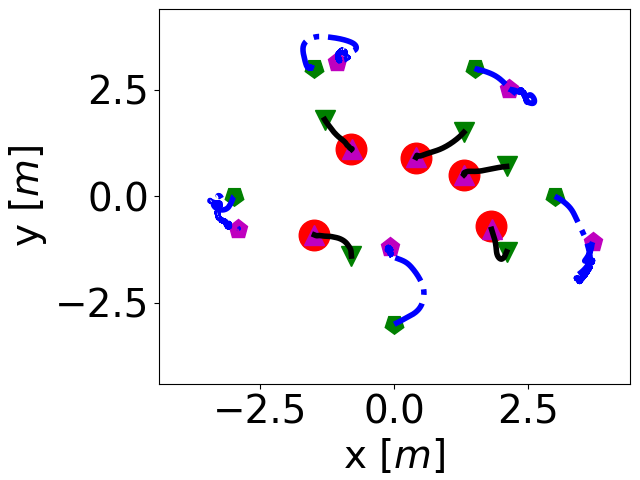}\hspace{0.1cm}
         &\hspace{-0.4cm}\hspace{0.1cm}
         \includegraphics[width=0.21\textwidth]{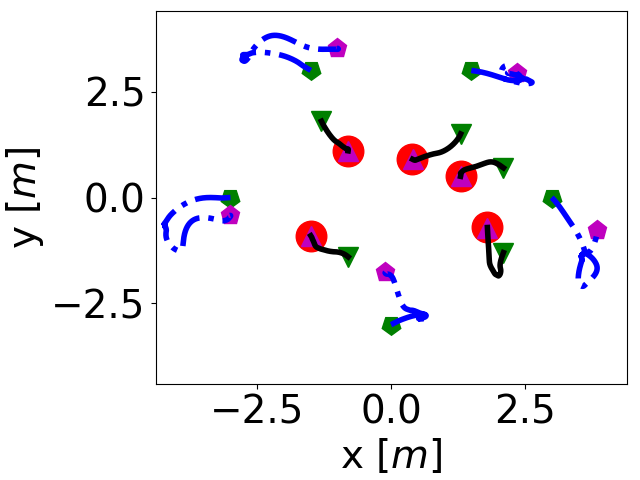}\hspace{0.1cm}
         &\hspace{-0.4cm}\hspace{0.1cm}
         \includegraphics[width=0.21\textwidth]{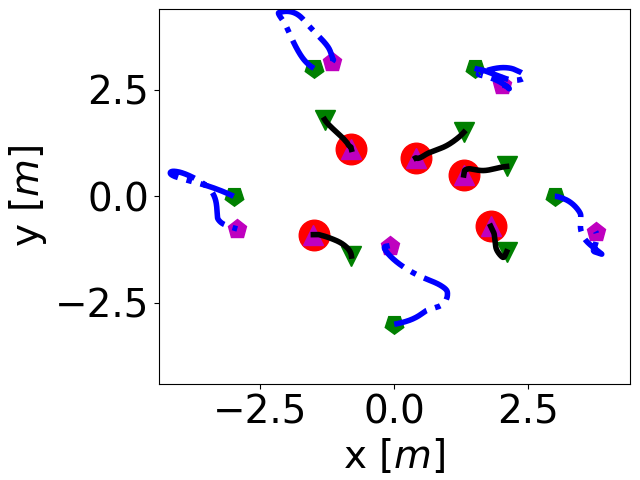}\hspace{0.1cm}
         \\\hspace{0.1cm}
         \includegraphics[width=0.21\textwidth]{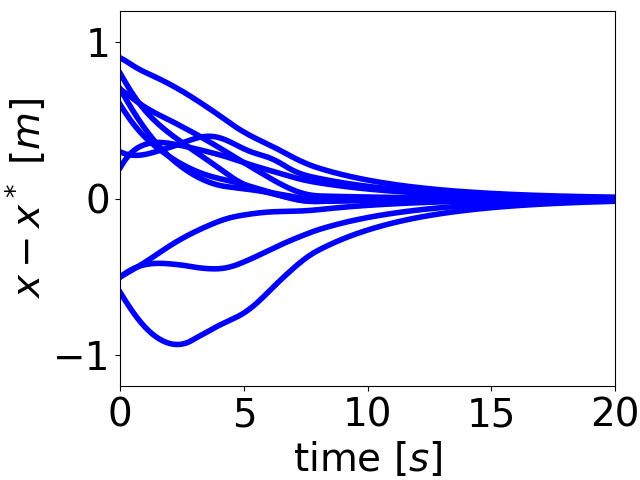}\hspace{0.1cm}
         &\hspace{0.1cm}
         \includegraphics[width=0.21\textwidth]{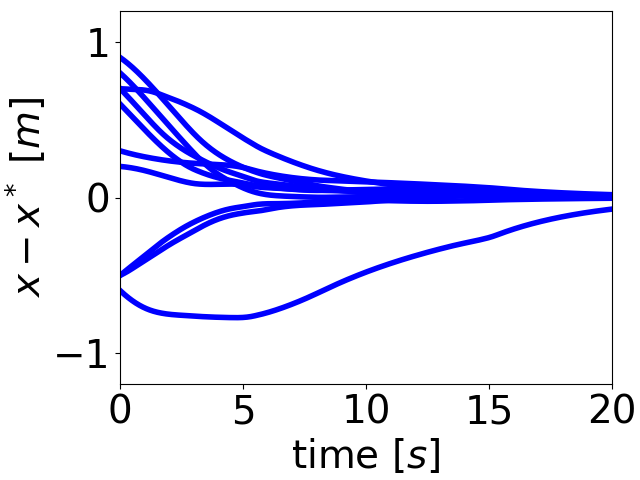}\hspace{0.1cm}
         &\hspace{0.1cm}
         \includegraphics[width=0.21\textwidth]{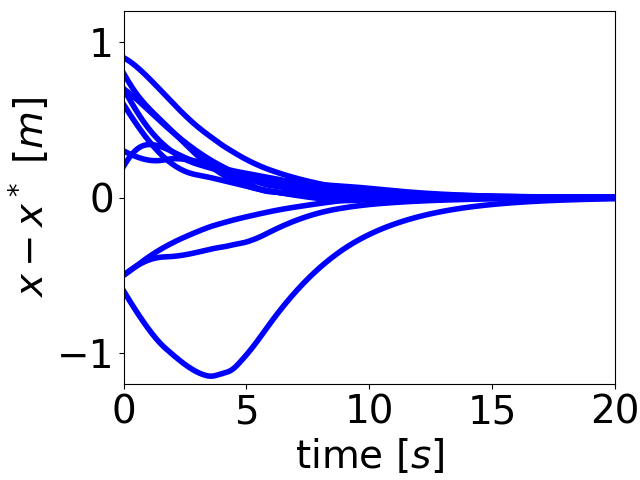}\hspace{0.1cm}
         &\hspace{0.1cm}
         \includegraphics[width=0.21\textwidth]{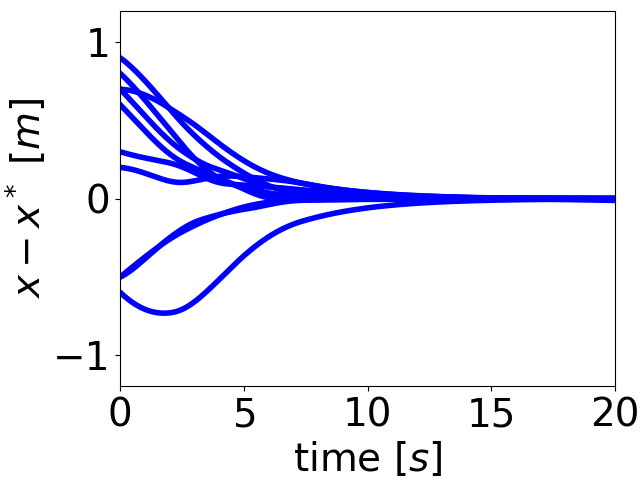}\hspace{0.1cm}
    \end{tabular}
	\caption{Simulation results of the herding of 5 evaders by 5 herders. The first row shows the trajectories followed by the herders and the evaders. The symbols are explained in Table~\ref{table:legend}. The second row presents the evolution of the error between desired and current position of the evaders.}
	\label{fig:sim_results}
	\vspace{-0.4cm}
\end{figure*}

\begin{figure}[!ht]
    \centering
    \begin{tabular}{cc}
         \hspace{0.2cm}{\footnotesize Inverse Model}
         &  
         \hspace{0.2cm}{\footnotesize Exponential Model}
         \\
        \includegraphics[width=0.46\columnwidth,height=0.30\columnwidth]{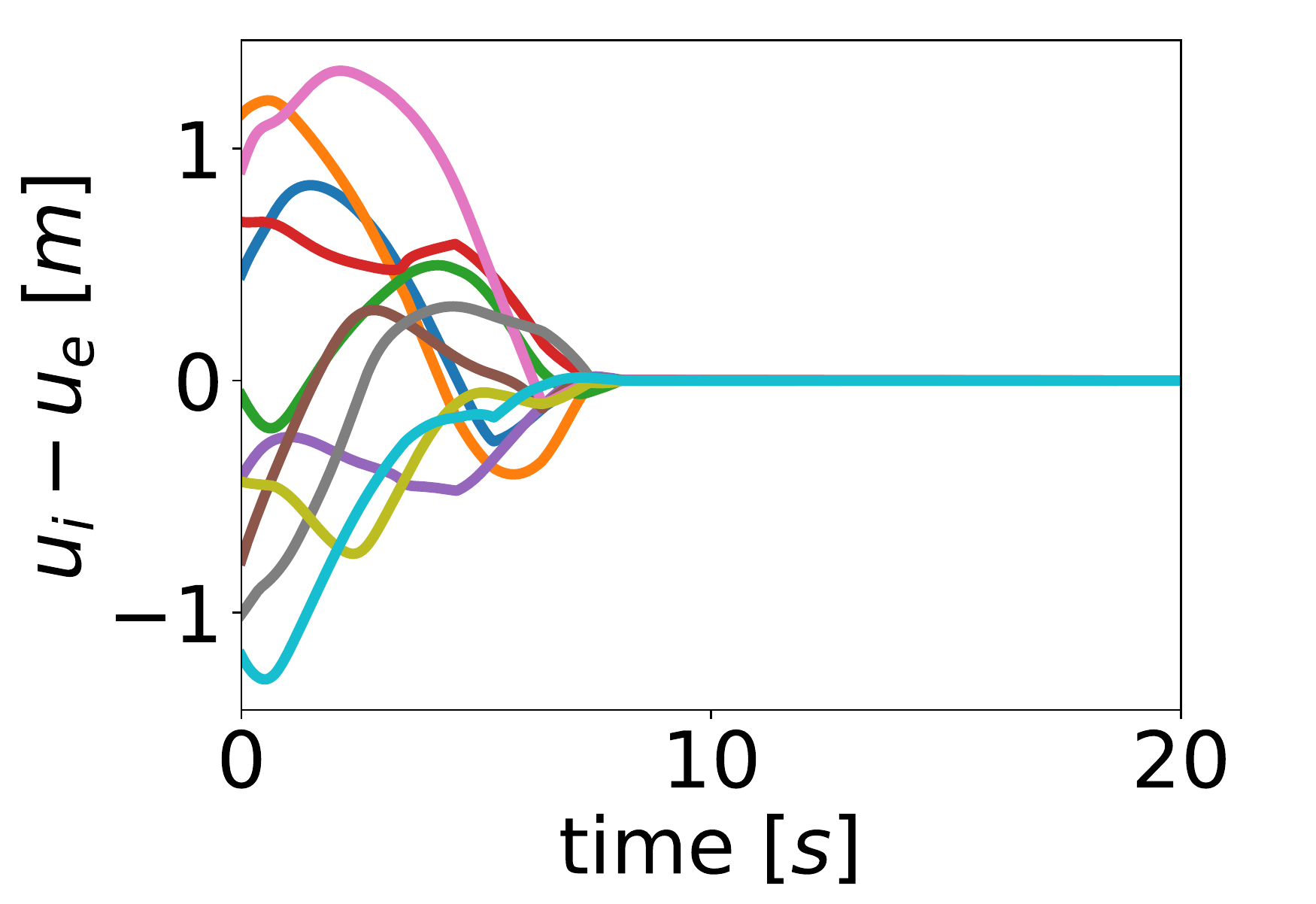}
         &  
        \includegraphics[width=0.46\columnwidth,height=0.30\columnwidth]{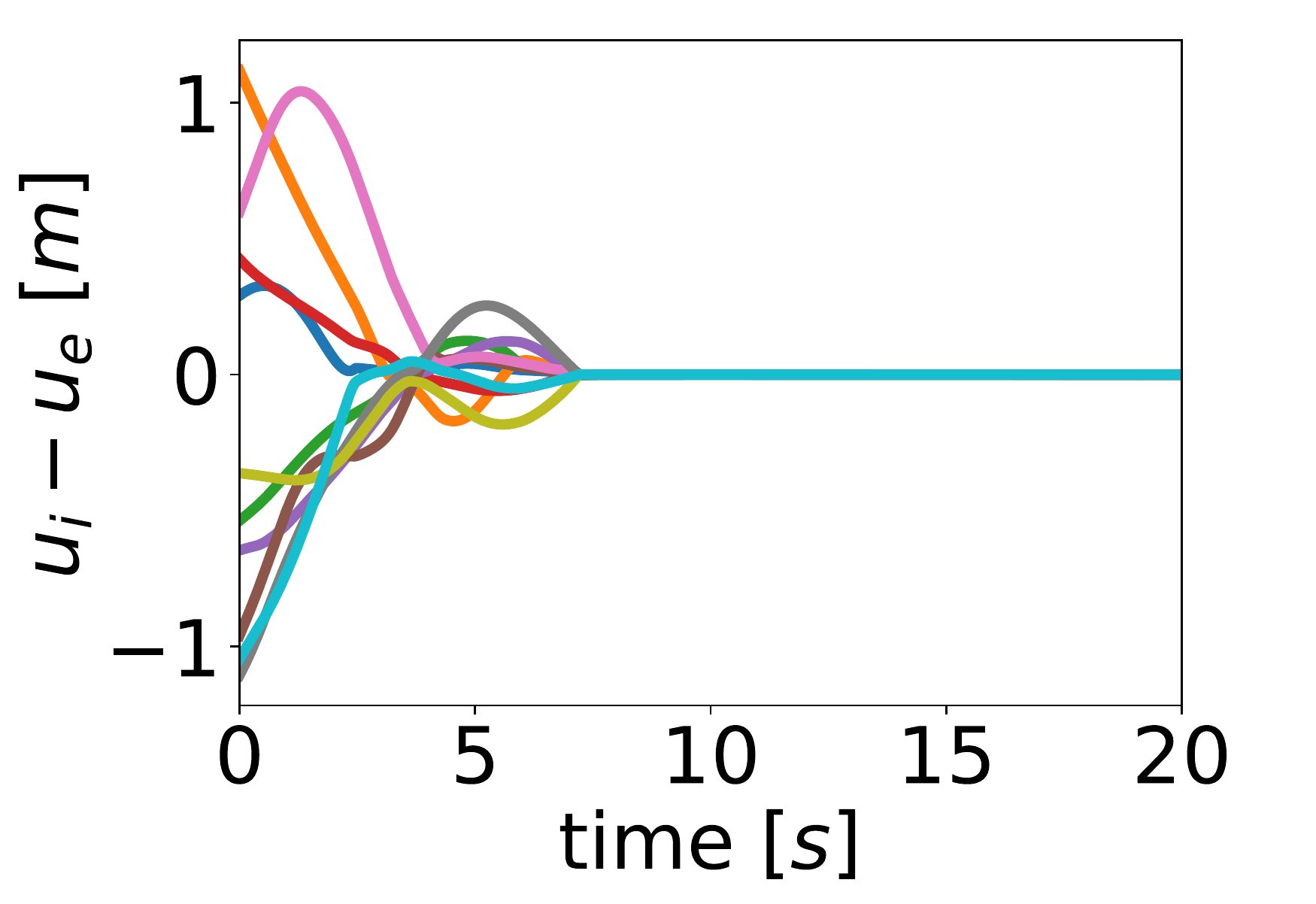}
    \end{tabular}
	\caption{Difference in the control input calculated by Implicit Control and the numerical baseline ($\mathbf{u}_i$ and $\mathbf{u}_e$ respectively). An arbitrary color has been assigned to each input for the sake of visibility.}
	\label{fig:diff_results}
\end{figure}

\begin{figure}[!ht]
    \centering
    \begin{tabular}{cc}
         \hspace{0.4cm}{\footnotesize Inverse evaders}
         &  
         \hspace{0.5cm}{\footnotesize Exponential evaders}
         \\
         \includegraphics[width=0.21\textwidth]{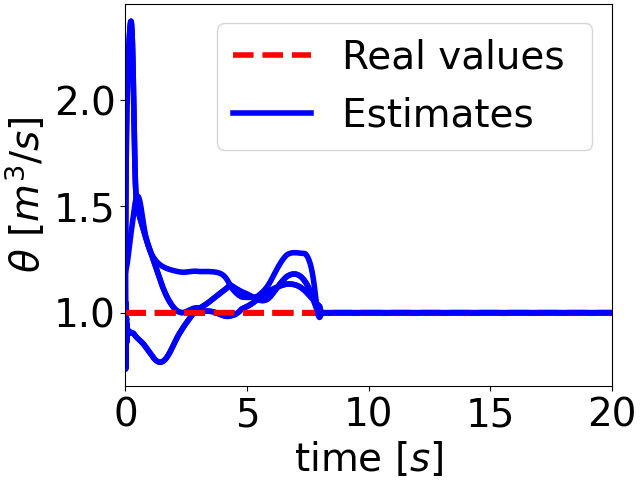}
         &
         \includegraphics[width=0.21\textwidth]{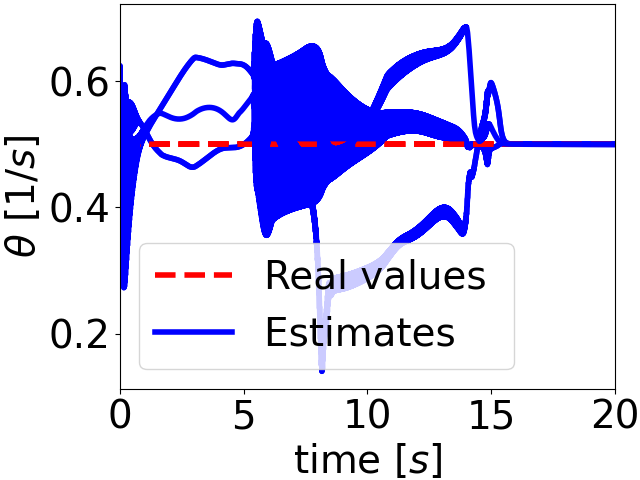}
    \end{tabular}
	\caption{Adaptation results of the herding of 5 evaders by 5 herders. The pannels depict the evolution of the estimated parameters.}
	\label{fig:adaptation_results}
\end{figure}

The only discrepancy between Implicit Control and the numerical baseline is in the first instants of the transient. The numerical baseline applies the ``correct input'' from the beginning, because at each instant it iterates until it finds the input that solves $f(\mathbf{x},\mathbf{u})-f^*(\mathbf{x}) = \mathbf{0}$.
Meanwhile, Implicit Control requires some time to converge to $f(\mathbf{x},\mathbf{u})-f^*(\mathbf{x}) = \mathbf{0}$, which is when the ``correct input'' is applied. Therefore, what only differs between methods is the shape of the transient in the first instants. Then, the total settling time with the Implicit Control is at most the settling time of $f^*(\mathbf{x})$ plus the time $h(\mathbf{x},\mathbf{u})$ takes to converge to zero. If $\mathbf{K}_h$ is large enough, then the latter is negligible. This is why the differences in the transient are minimal in the second row of Fig.~\ref{fig:sim_results}. Meanwhile, the differences in the input are much greater, as it is seen in Fig.~\ref{fig:diff_results}. The controlled system is strongly input-nonaffine, so there are many roots of $h$ that solve $h(\mathbf{x},\mathbf{u})=\mathbf{0}$. After Implicit Control converges to $h(\mathbf{x},\mathbf{u}) = \mathbf{0}$, herders are in a different configuration with respect to the herders driven by the numerical baseline, and this difference implies different roots for $h(\mathbf{x},\mathbf{u})$ until the evaders are close to the desired configuration.

Fig.~\ref{fig:adaptation_results} shows how the adaptation law adjusts the parameters to ensure stability. Initially, all the parameters are equal; however, they evolve differently because each evader is in a specific situation with respect to its desired position and the herders. At steady-state, the estimated parameters are equal to the real ones.

\subsection{Effects of the E-DKF in the herding}

The first series of simulations has validated the success of Implicit Control and the adaptation law. Next, we show how the herding works when perfect state and input knowledge is replaced by the distributed estimator. The communication and sensing thresholds are in Table~\ref{table:initial_values}. We set $\mathbf{Q} = 0.02\mathbf{I}_{2m+2n}$ and $\mathbf{R}_i = 0.07\mathbf{I}_{2m+2n}$ $\forall i$ to simulate a scenario with bad sensing capabilities, i.e., $\mathbf{R}_i \succ \mathbf{Q}$ to illustrate how incorporating Implicit Control in the prediction stage is useful. Finally, to be realistic in the use of the communication bandwidth, messages are exchanged every $100$ms, i.e., the communication process runs $10$ times slower than the control. Fig.~\ref{fig:sim_dekf} depicts the performance of the herders against the same herding scenario of Fig.~\ref{fig:sim_results}. The results are almost identical, with herders and evaders following similar trajectories as if there was perfect feedback. The reason for that is answered in the second row of Fig.~\ref{fig:sim_dekf}. The Root Mean Square Error (RMSE) between the estimates and real positions of the entities rapidly decreases to the noise level. We remark that both axis of the panels are in logarithmic scale. Thus, before the first $400$ms ($4$ communication rounds) the estimator has converged, which justifies the assumption of perfect knowledge in the control and adaptation. Despite the different frequencies between communication and control, the herders are successful because they can predict the movement of the evaders \textit{and} the other herders. 

\begin{figure}[!ht]
    \centering
    \begin{tabular}{cccc}
         \hspace{0.5cm}{\footnotesize Inverse evaders}
         &  
         \hspace{0.9cm}{\footnotesize Exponential evaders}
         \\\hspace{-0.35cm}\hspace{0.1cm}
         \includegraphics[width=0.19\textwidth]{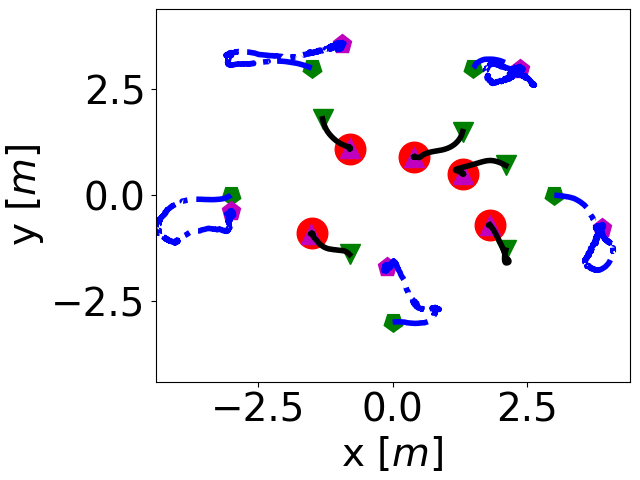}\hspace{0.1cm}
         &\hspace{-0.1cm}\hspace{0.1cm}
         \includegraphics[width=0.19\textwidth]{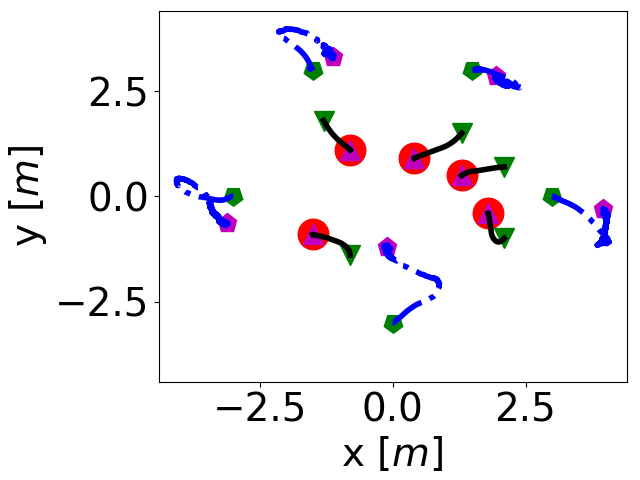}\hspace{0.1cm}
         \\\hspace{-0.35cm}\hspace{0.1cm}
         \includegraphics[width=0.19\textwidth]{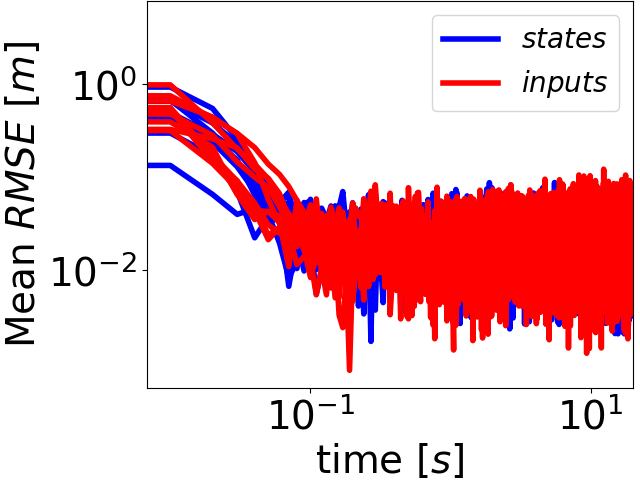}\hspace{0.1cm}
         &\hspace{0.1cm}
         \includegraphics[width=0.19\textwidth]{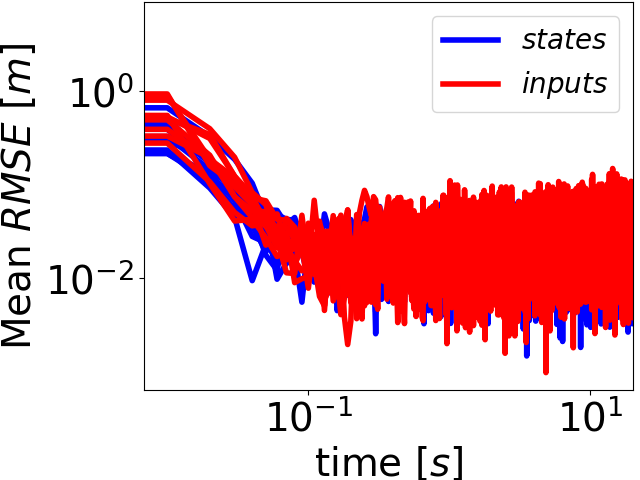}\hspace{0.1cm}
    \end{tabular}
	\caption{Simulation results of the herding of 5 evaders by 5 herders using Implicit Control with adaptation law and the distributed estimator. The first row shows the trajectories followed by the herders (same symbols of Table~\ref{table:legend}). The second row depicts the RMSE in the estimation of evaders' (blue) and herders' (red) position, averaged over the five herders.}
	\label{fig:sim_dekf}
\end{figure}

In this vein, it is interesting to analyze how increasing the noise covariance affects the performance. To do so, we have run the same simulations in Fig.~\ref{fig:sim_dekf} but changing the measurement covariance matrix, checking to what extent the noise can increase without affecting the performance. The results are in Table~\ref{table:noises}. Implicit Control is robust against measurement noise, achieving good performance with noises up to $0.3$m$^2$. We underline that this implies measurement errors in the order of $0.5$m, which is significant considering that the herding takes place in a $5\times5$m square (i.e., a $10\%$ of the side size of the arena). It is also observed that the tolerance to noise depends on the evaders' dynamics. On the other hand, we do not consider changes in $\mathbf{Q}$ because the movement of herders and evaders is deterministic, so $\mathbf{Q}$ just models the noise associated with the difference between real and estimated parameters during the adaptation.

\begin{table}[!ht]
\centering
\caption{Impact of the measurement noise in the herding.}
\begin{tabular}{|c|c|c|}
\hline
 \multirow{2}{*}{$\mathbf{R}$ [m$^{2}$]} & \multicolumn{2}{c|}{Steady-state error $||\mathbf{x}-\mathbf{x}^*||$ [m]} \\ \cline{2-3} & Inverse Model, $5$vs$5$ & Exponential Model, $5$vs$5$ \\
\hline
$0.07 \mathbf{I}_{20}$  &  $0.02$ & $0.01$
\\
\hline
$0.15 \mathbf{I}_{20}$  &  $0.06$ &  $0.06$ \\ 
\hline
$0.30 \mathbf{I}_{20}$  &  $1.34$ &  $0.12$ \\ 
\hline
\end{tabular}
\label{table:noises}
\end{table}

\vspace{-0.4cm}

\subsection{Complete precise herding}

The flexibility and generality of the solution can be extended to heterogeneous groups and time-varying references, resulting in a more realistic herding. In this example we evaluate the inclusion of the caging stage as considered in  Algorithm~\ref{al:with_Ad}. Fig.~\ref{fig:TV_example} shows how three herders herd a group of three evaders using the complete herding solution in Algorithm~\ref{al:with_Ad}. The red evader is Exponential while the purple evaders are Inverse. The CBFs have been tuned as follows, choosing
\begin{equation}\label{eq:h_cbf}
    h(\Delta_i) = ||\Delta_i|| + T \frac{\Delta_i \delta\mathbf{u}_i^{nom}}{||\Delta_i||} - \varphi
\end{equation}
and
\begin{equation}\label{eq:kappa_cbf}
    \kappa(h(\Delta_i)) = k_{cbf} h(\Delta_i)^3, 
\end{equation}
where $T$ is the sampling time, $\varphi=3$ is a desired distance with the perimeter of $\mu_2 \mathbf{P}_{avg}$ ($\mu_1=7$ and $\mu_2=4$) and $k_{cbf} = 50$ is a gain which modulates how strong is the repulsion provoked by the proximity with $\mathbf{e}_i$. We recall that in the caging stage the dynamics of the herders are
\begin{equation}\label{eq:update_cbf}
    \mathbf{u}_i = \mathbf{u}_i + T \delta \mathbf{u}_i^{cbf}.
\end{equation}

Meanwhile, the desired herding configuration evolves with
\begin{align*}
    \dot{x}^*_j = v_j^* ,  
    & &
    \dot{y}^*_j = 0.5 w_j^* \cos(w_j^* t + 2 \pi / j),
\end{align*}
where $\mathbf{w}^* \!\!=\!\! [0.05, 0.1, 0.02]$rad/s and $\mathbf{v}^*\!\! =\!\! [0.05, 0.05, 0.05]$m/s. Due to the distances involved in this herding instance, we set $d_m = d_c = 15$m.

Initially, the herders move to uniformly surround the evaders, positioning themselves on the perimeter of the desired ellipse. Once the herders are there, the precise herding begins, driving the evaders to their sinusoidal references. This yields to herders' trajectories surrounding and modulating the interaction forces with the evaders. With the evaders in their desired trajectories, the system reaches a steady-state behavior where the periodic movement of the evaders is shared by the herders. The clip of this simulation is included in the supplementary material. 
\begin{figure*}[!ht]
    \centering
    \begin{tabular}{ccccc}
        \hspace{0.5cm}{\footnotesize Global}
        &
        \hspace{0.65cm}{\footnotesize Caging}
        &
        \hspace{0.8cm}{\footnotesize Meantime snapshot}
        &
        \hspace{0.5cm}{\footnotesize Following desired trajectories}
        \\
        \includegraphics[width=0.22\textwidth,height=0.15\textwidth]{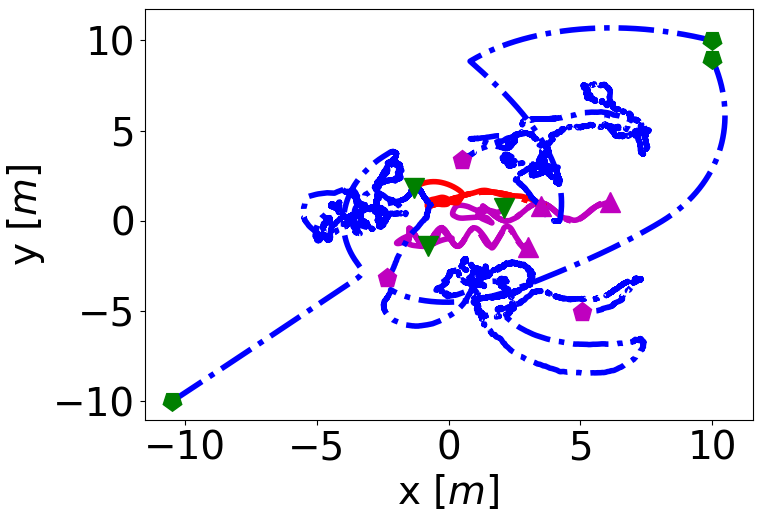}
        &
        \includegraphics[width=0.22\textwidth,height=0.15\textwidth]{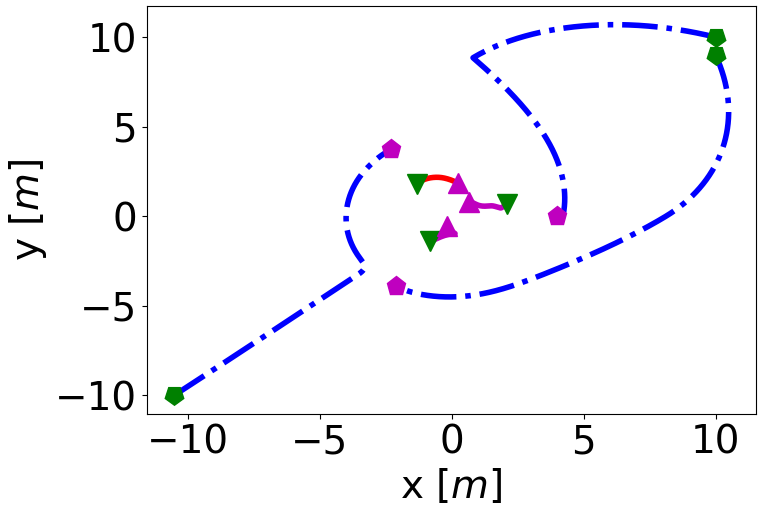}
        &
        \includegraphics[width=0.22\textwidth,height=0.15\textwidth]{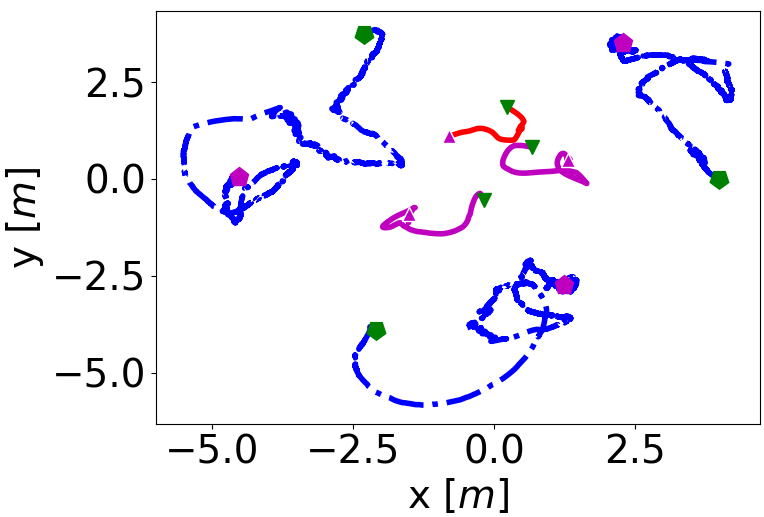}
        &
        \includegraphics[width=0.22\textwidth,height=0.15\textwidth]{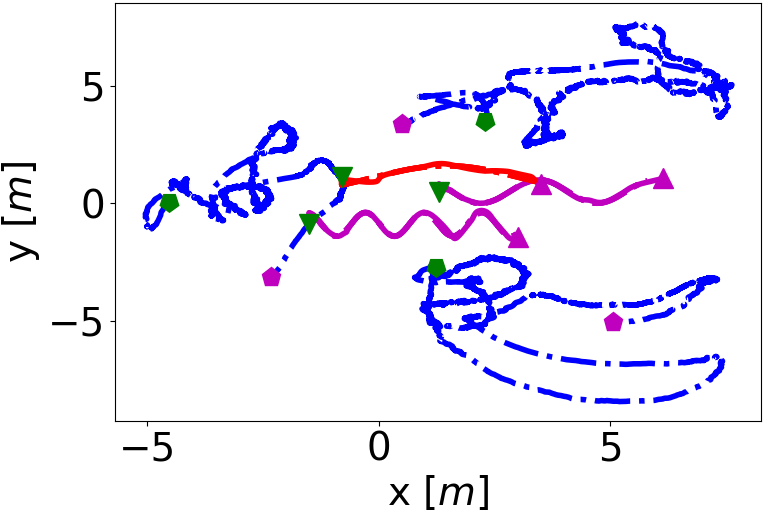}
    \end{tabular}
	\caption{Three robotic herders herding three heterogeneous evaders. The magenta trajectories correspond to Inverse evaders while the red trajectory is an Exponential evader. The other symbols follow the convention in Table~\ref{table:legend}.}
	\label{fig:TV_example}
\end{figure*}

\subsection{Controlling large herds with a few herders}

To demonstrate the generality and flexibility of Implicit Control, we conduct some experiments where Implicit Control is used to control a large number of evaders with a few herders. To do this, the control goal is to steer the \textit{centroid} of the herd towards a particular location. At each instant, the control algorithm first calculates the centroid of the herd and, then, it computes the input with Implicit Control, using the centroid as if it was a virtual evader with its corresponding evader's dynamics. The computation of the centroid induces a little noise in the dynamics of the virtual evader, but it is negligible because, point-wise, the centroid's dynamics is still that of the virtual evader. Besides, this discrepancy can be absorbed by the adaptation law. 

In Fig.~\ref{fig:sim_results3} it is shown how $5$ herders can control $50$ evaders. There are $10$ inputs and $2$ states to control (the $x$ and $y$ position of the centroid). 
To maintain cohesion of the herd, but also to avoid collisions among the evaders, we have added very weak repulsive and coalition forces among evaders
\begin{equation}\label{eq:additional_forces}
    \dot{\mathbf{x}}_j \kern -0.1cm = \kern -0.1cm f_j^{inv/exp}(\mathbf{x},\mathbf{u}) + \vartheta \sum_{j^{\prime}=1}^{m} \mathbf{d}_{jj^{\prime}}\left( \frac{1}{||\mathbf{d}_{jj^{\prime}}||^{3}} - ||\mathbf{d}_{jj^{\prime}}||^2\right) 
\end{equation}
with $\mathbf{d}_{jj^{\prime}} = \mathbf{x}_j - \mathbf{x}_{j^{\prime}}$ and $\vartheta = 2 \times 10^{-4}$, i.e., $4$ orders of magnitude smaller than the dynamics in Eqs.~\eqref{eq:PiersonBase} and~\eqref{eq:LicitraBase}.

The first row of Fig.~\ref{fig:sim_results3} shows how the five herders can steer the whole herd towards a certain region. As a curiosity, since the number of available control inputs is much greater than the number of states to control, there is a ``leftover'' herder in steady-state (right pannel):  there are two herders which share a similar position. More interestingly, the second row of Fig.~\ref{fig:sim_results3} shows how the herders can split the herd in two sub-herds and steer each of them, simultaneously, towards individually assigned locations. This confirms that, with our herding, it is possible to consider each evader to virtually represent a herd (see, e.g.,~\cite{Pierson_2018_TR_Herding} for a similar consideration).

\begin{figure*}[!ht]
    \centering
             \includegraphics[width=1\textwidth]{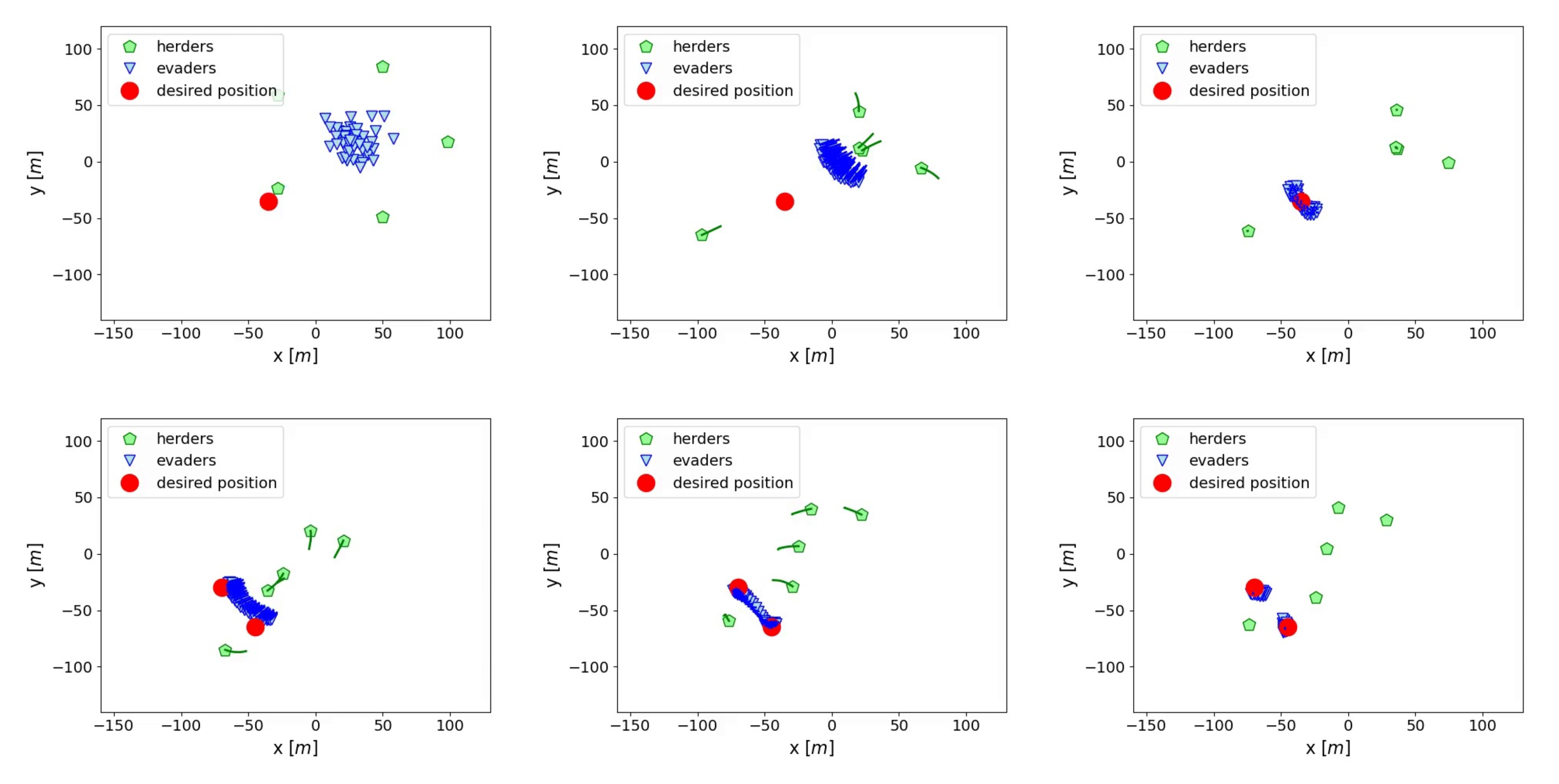}
	\vspace{-0.8cm}
	\caption{$5$ herders herding $50$ Inverse evaders. In the first $50$ seconds, the herd is steered towards the desired location (first row). Afterwards, the herd is split in two sub-herds and driven towards two different desired locations simultaneously (second row).}
	\vspace{-0.4cm}
	\label{fig:sim_results3}
\end{figure*}

The supplementary material includes a video with different configurations of evaders, demonstrating that Implicit Control can herd large heterogeneous groups of evaders while mixing Inverse and Exponential evaders.


\section{Experiments}\label{sec:experiments}

In this Section we extend the experiments to the real framework provided by the Robotarium~\cite{Wilson2020Robotarium}. To do so, some robots play the role of herders while the others act as evaders, following the dynamics in Section~\ref{sec:prosta}. The robots are GRITSBot X moving on a $3.2$m x $2$m area, coordinated by a central server which receives odometry and sends velocity commands at an approximately delay of $0.033$s. Thus, a low level controller is used to translate the control output into velocity commands, with $v_{max}=0.2$m/s. We use the experiments to validate Implicit Control, the adaptation law and the caging stage. Finally, we adjust some parameters to fit the conditions of the experiment: $T=0.033$s, $\theta = 0.02,$ $\hat{\theta}(0)=0.015$ (Inverse), $\theta=0.05,$ $\hat{\theta}(0)=0.04$ (Exponential), $\sigma = 1.2$. 
\begin{figure*}[!ht]
\centering
\begin{tabular}{cccc}
    {\footnotesize 3 herders vs 3 Inv. evaders} 
    & {\footnotesize 3 herders vs 3 Exp. evaders} 
    & {\footnotesize 3 herders vs 2 Inv. $+$ 1 Exp. evaders} 
    & {\footnotesize 4 herders vs 2 Inv. $+$ 2 Exp. evaders}
    \\\hspace{0.1cm}
    \includegraphics[width=0.21\textwidth]{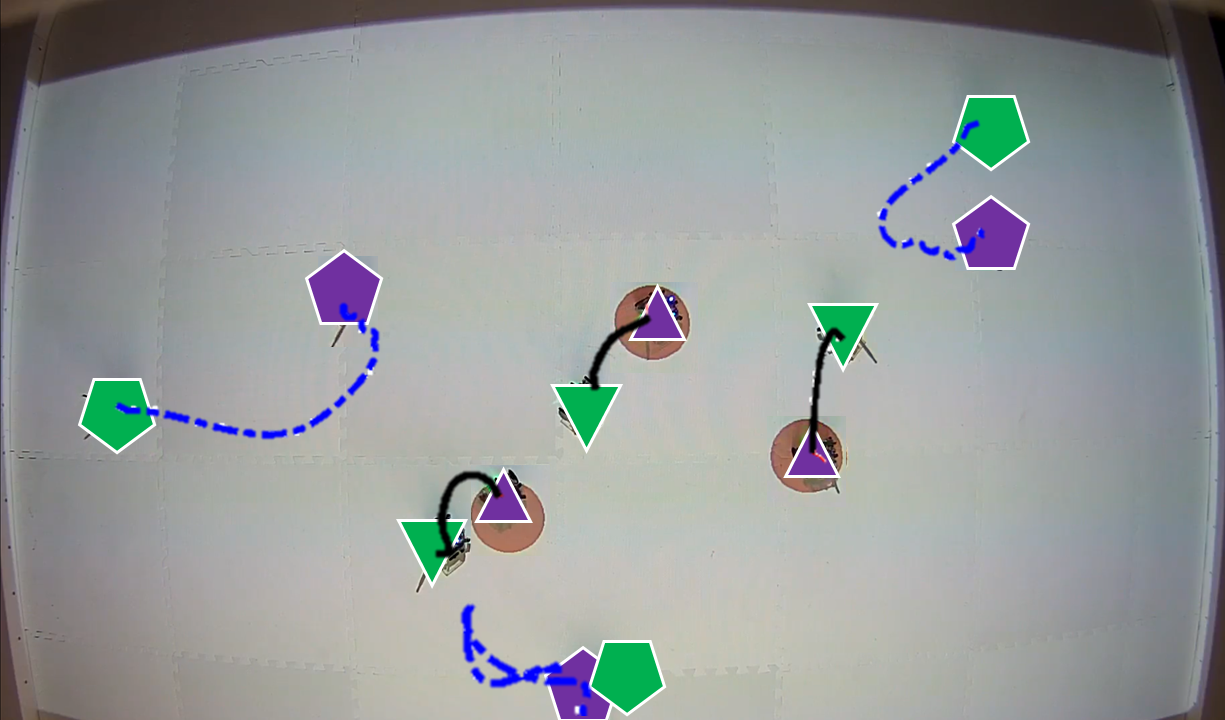}\hspace{0.1cm}
    & \hspace{0.1cm}
    \includegraphics[width=0.21\textwidth]{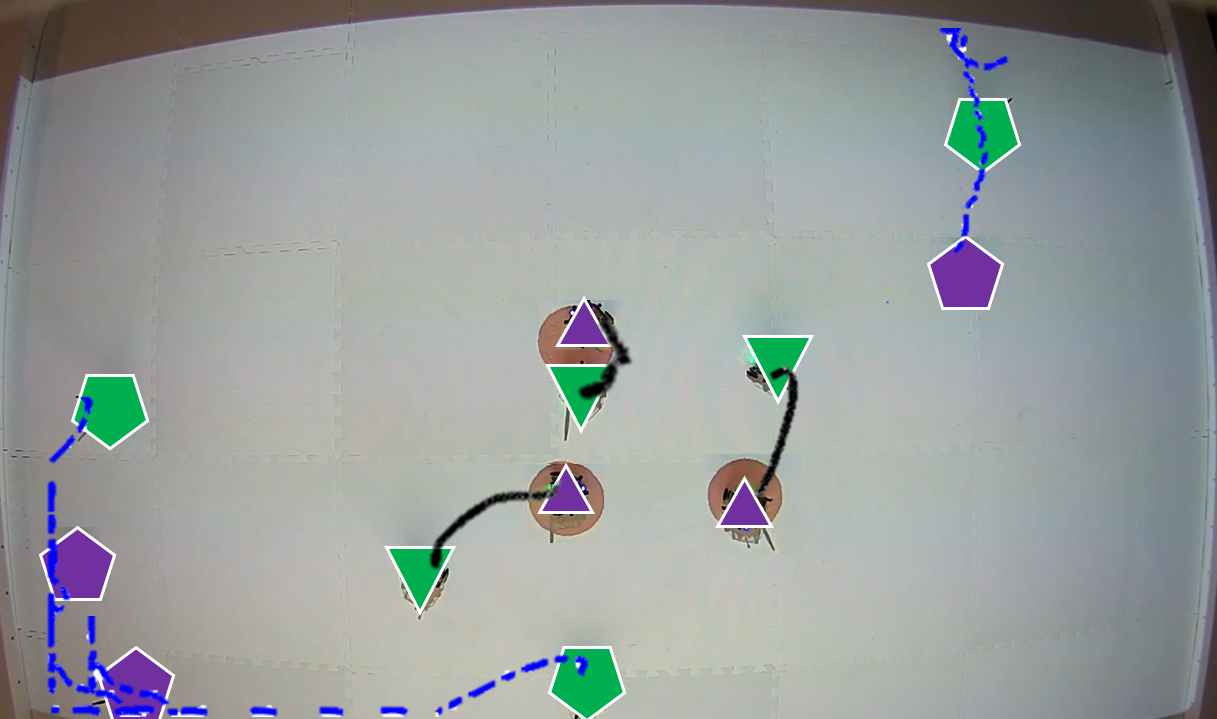}\hspace{0.1cm}
    &\hspace{0.1cm}
    \includegraphics[width=0.21\textwidth]{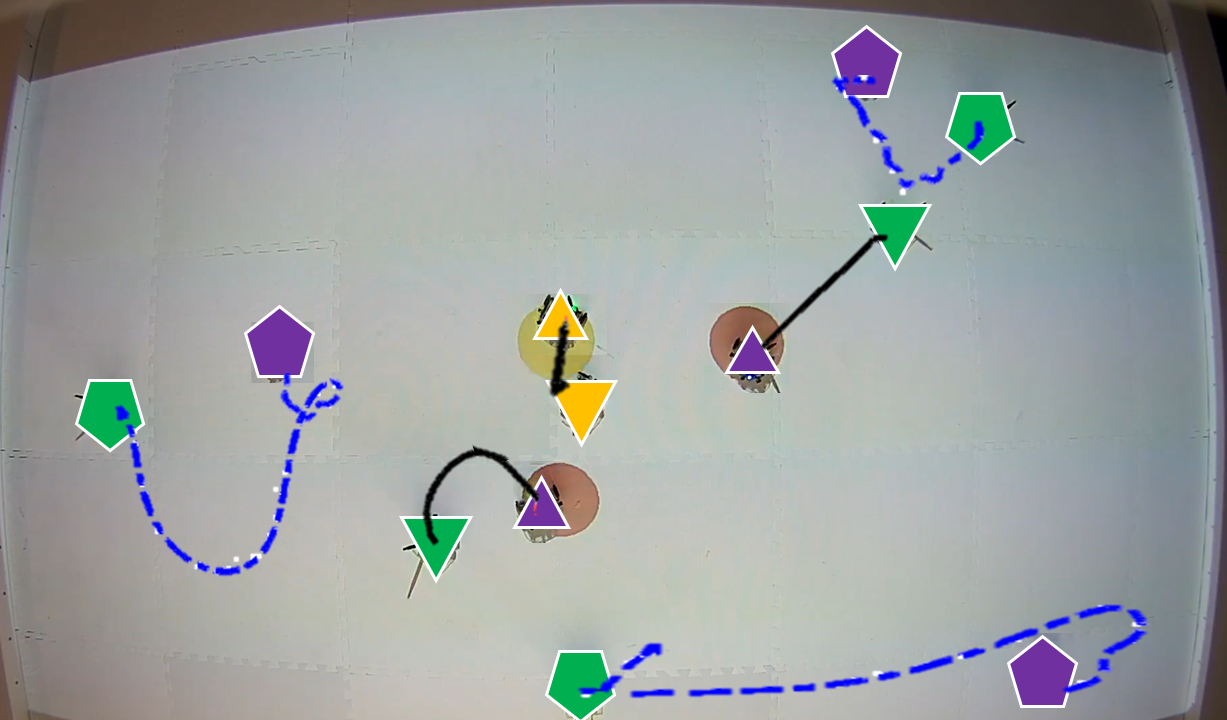}\hspace{0.1cm}
    &\hspace{0.1cm}
    \includegraphics[width=0.21\textwidth]{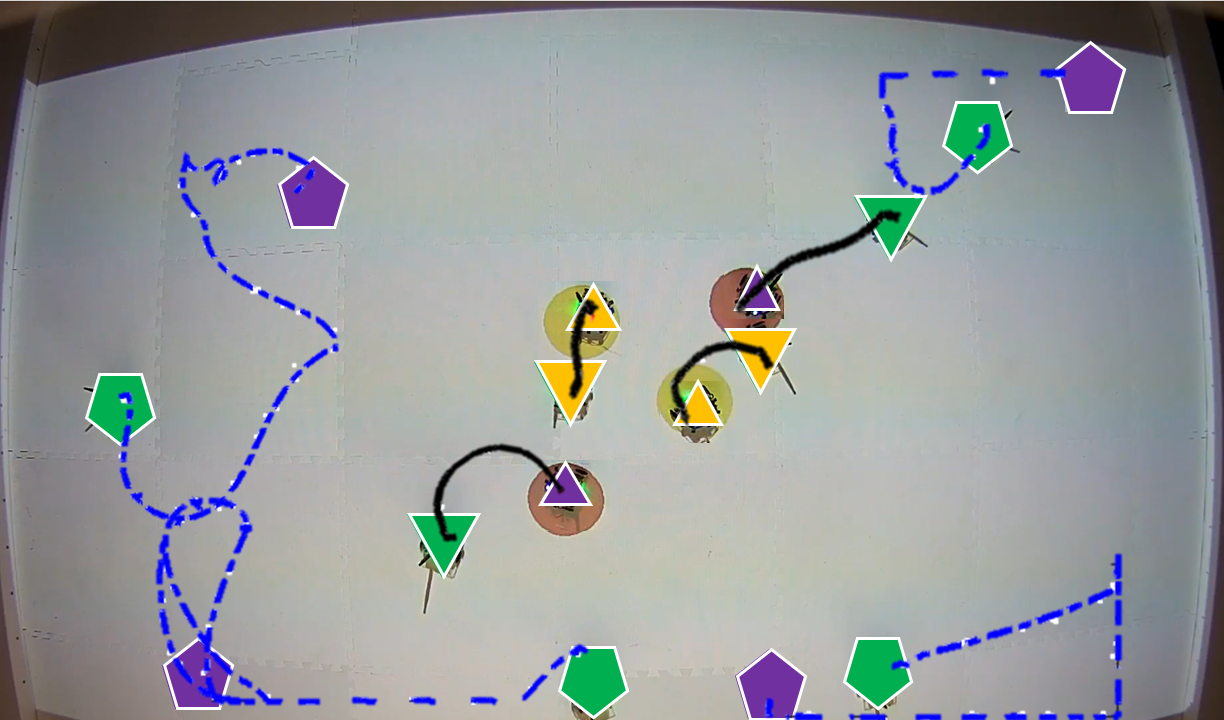}\hspace{0.1cm}
    \end{tabular}
\caption{Results of the herding using Robotarium~\cite{Wilson2020Robotarium}. The symbols are explained in Table~\ref{table:legend}. When evaders' dynamics are mixed, initial, final and desired position of Exponential evaders are yellow. A video with the complete experiments is included as supplementary material.
}
\label{fig:second_impression}
\end{figure*}
\begin{figure*}[!ht]
\centering
\begin{tabular}{cccc}
    {\footnotesize Initial configuration} 
    & {\footnotesize Caging} 
    & {\footnotesize Meantime snapshot} 
    & {\footnotesize Final trajectories}
    \\\hspace{0.1cm}
    \includegraphics[width=0.23\textwidth]{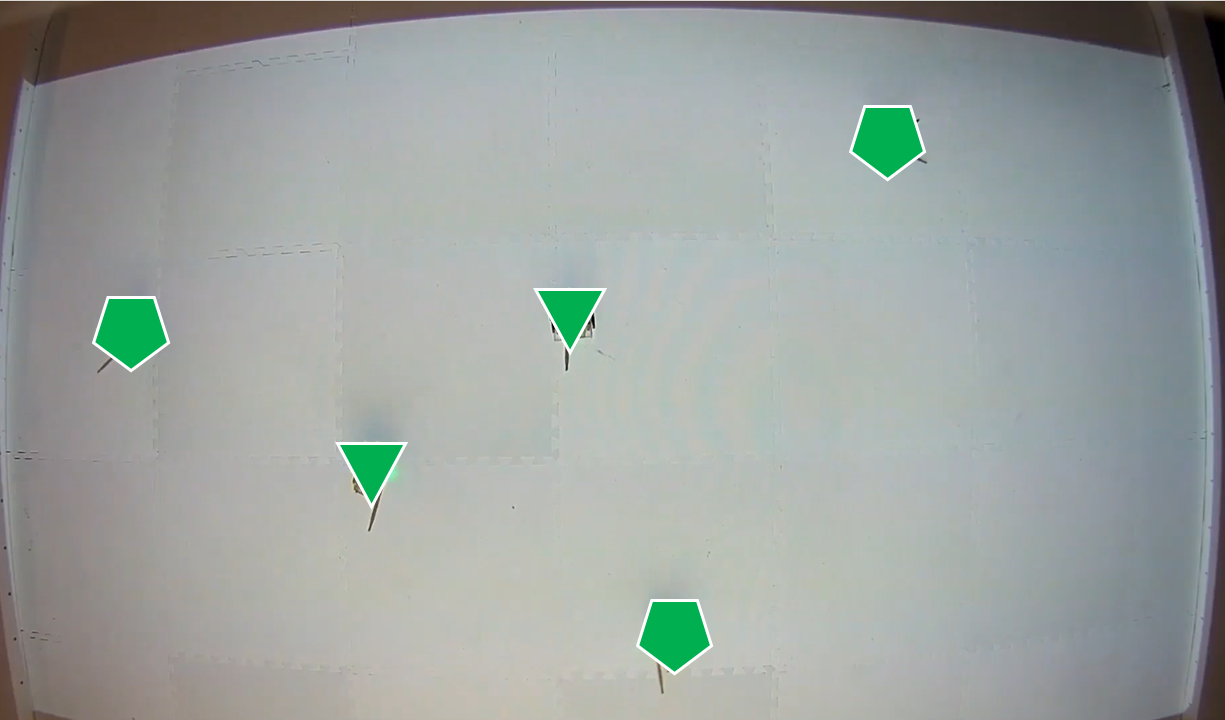}\hspace{0.1cm}
    & \hspace{0.1cm}
    \includegraphics[width=0.22\textwidth]{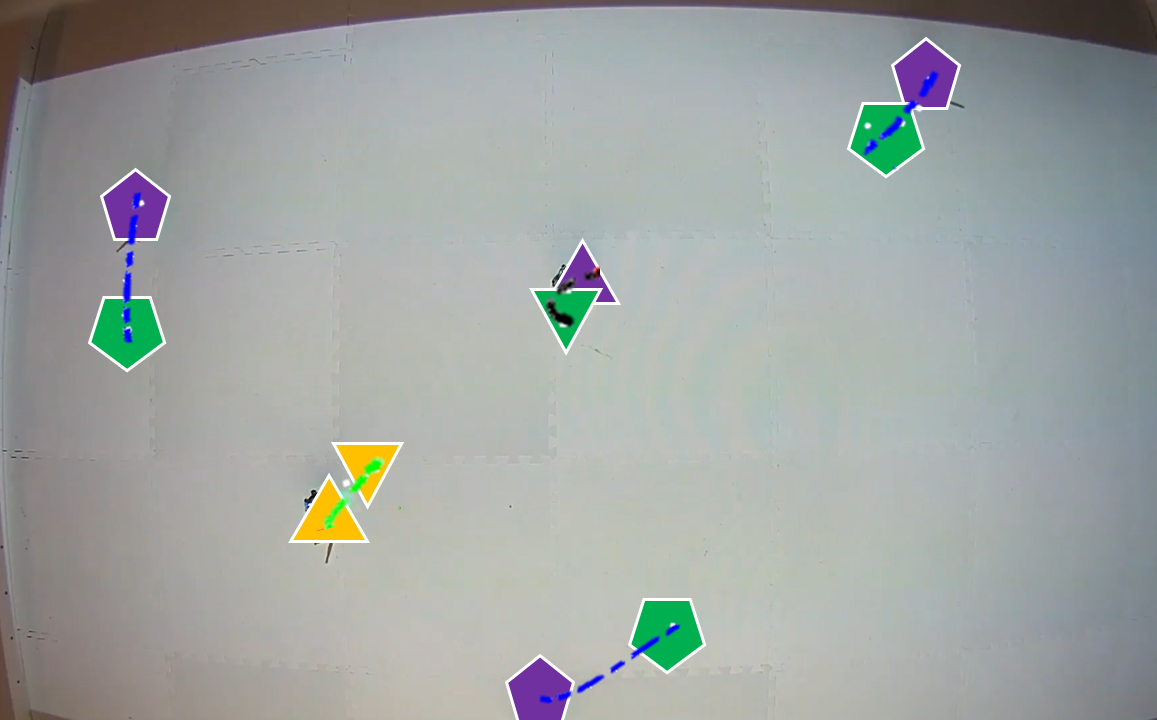}\hspace{0.1cm}
    &\hspace{0.1cm}
    \includegraphics[width=0.21\textwidth]{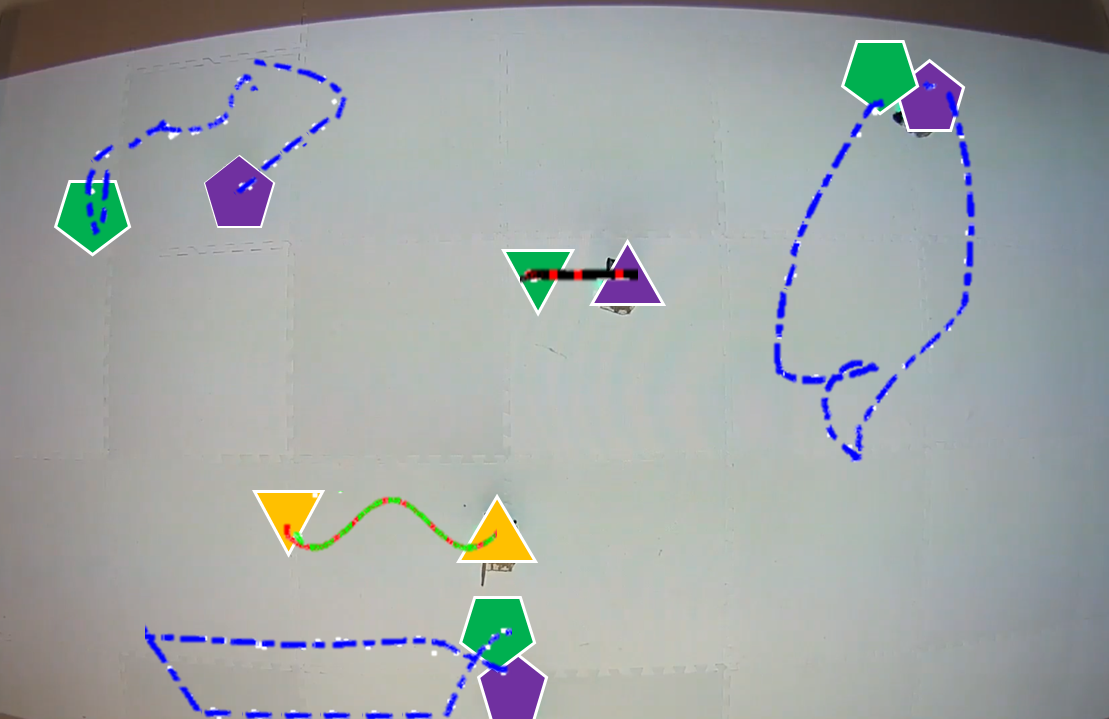}\hspace{0.1cm}
    &\hspace{0.1cm}
    \includegraphics[width=0.21\textwidth]{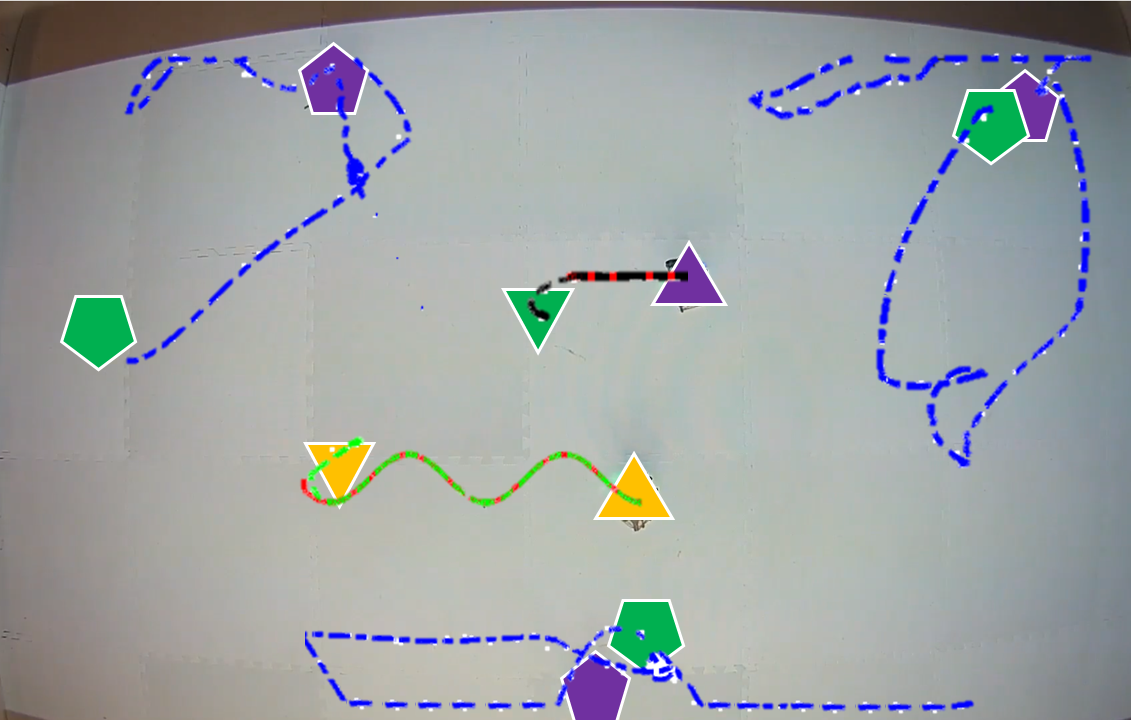}\hspace{0.1cm}
    \end{tabular}
\caption{Three herders herding two heterogeneous evaders in a time-varying path. More details in the supplementary video.
}
\vspace{-0.2cm}
\label{fig:robotarium_with_approaching}
\end{figure*}

The experiment in Fig.~\ref{fig:first_impression} exhibits a similar behavior to the simulations in Section~\ref{sec:simulations}. The evaders try to evade the herders, going in the direction of lower density of herders. To tackle this, the closest herders surround the evaders to align with the other herders, which move away to modulate the interaction forces. These conclusions are reaffirmed by the experiments in Fig.~\ref{fig:second_impression}, where different combinations of evaders, in their number and dynamics, are tested. 
Additionally, the herders successfully herd heterogeneous groups of evaders, as it is demonstrated in Fig.~\ref{fig:second_impression}c, Fig.~\ref{fig:second_impression}d, and in the time-varying experiment in Fig.~\ref{fig:robotarium_with_approaching}. 

The latter consists in the herding of an Inverse evader and an Exponential evader by three herders, where the caging stage has been included. Initially, the herders approach the evaders, surrounding them to avoid escapes. Then they steer the evaders towards the desired trajectories (in red). Despite the space limitations and the complex nonlinear repulsive dynamics, the evaders successfully follow the references. 

As a detail, in the video it is observed that, for the case of $4$ evaders vs $4$ herders, the two herders on the left side of the frame exchange their positions at a certain instant. This is because the control input is defined in terms of its dynamics, so it ``follows the flow of a numerical solver''. Therefore, herders may exchange their position because this is how the input dynamics is evolving towards the desired dynamics. To avoid this behavior in a particular application, a low-level controller with an obstacle-avoidance mechanism can be incorporated. Anyway, this is not a concern of Implicit Control but of the particular implementation in the application.

All the experiments are successful despite uncertainty due to the adaptation law. The supplementary video includes the complete clips of all the experiments.


\section{Conclusions}\label{sec:conclusion}

This paper has addressed a novel control strategy to solve the herding problem in MRS. This strategy, based on numerical analysis theory and coined as Implicit Control, finds suitable herding inputs even when, due to the complex nonlinearities of the herd, the control law is given by a set of implicit equations. Implicit Control develops a continuous-time expansion of the system, and comes with formal proofs of convergence. In addition, it is flexible in the number of evaders and general with respect to their motion model. The theoretical concepts of Implicit Control allow to derive an adaptation law. It preserves the stability properties of the control while naturally integrates in the control architecture.
To complete the proposal, we have developed a novel solution for the caging problem. The herders can depart from any arbitrarily far region because they are able to approach and surround the evaders. This is done in conjunction with a novel version of an Extended Distributed Kalman Filter, which has solved the problem of complete perfect measurements. The flexibility, robustness, and generality of the solution has been demonstrated in numerous simulations and experiments with real robots.

Implicit Control paves the way for novel research lines of potential interest, including the full distribution of the control computation, e.g., using distributed optimization techniques. Another interesting direction is to add learning tools to infer online not only linear-dependent parameters but also the nonlinear structure of the herd dynamics.

\bibliographystyle{IEEEtran}
\bibliography{IEEEabrv,IEEEexample.bib}


\vspace{-6.2cm}

\begin{IEEEbiography}[{\includegraphics[width=1in,height=1.25in,clip,keepaspectratio]{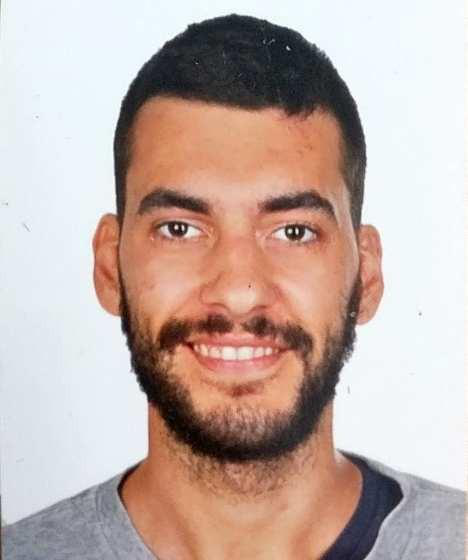}}]
{Eduardo Sebasti\'{a}n} (S'21) received the B.Eng. in Electronic and Automatic Engineering (Hons.) and the M.Eng. in Electronics (Hons.)  from the Universidad de Zaragoza in 2019 and 2020 respectively. He is currently a Ph.D. Candidate in the  Departamento de Informatica e Ingenieria de Sistemas at the Universidad de Zaragoza, funded by a FPU national grant (1st rank). His research interests are nonlinear control, distributed systems and multi-robot systems.
\end{IEEEbiography}

\vspace{-6.2cm}

\begin{IEEEbiography}[{\includegraphics[width=1in,height=1.25in,clip,keepaspectratio]{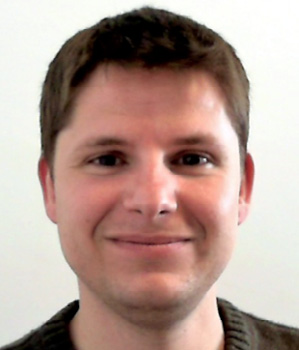}}]{Eduardo Montijano}
(M'12) is an Associate Professor in the Departamento de Inform\'atica e Ingenier\'ia de Sistemas at Universidad de Zaragoza in Spain. He received the M.Sc. and Ph.D. degrees from the Universidad de Zaragoza, Spain, in 2008 and 2012 respectively. He was a faculty member at Centro Universitario de la Defensa, Zaragoza, between 2012 and 2016. His main research interests include distributed algorithms and automatic control in perception problems. His Ph.D. obtained the extraordinary award of the Universidad de Zaragoza in the 2012-2013 academic year.
\end{IEEEbiography}

\vspace{-6.2cm}

\begin{IEEEbiography}[{\includegraphics[width=1in,height=1.25in,clip,keepaspectratio]{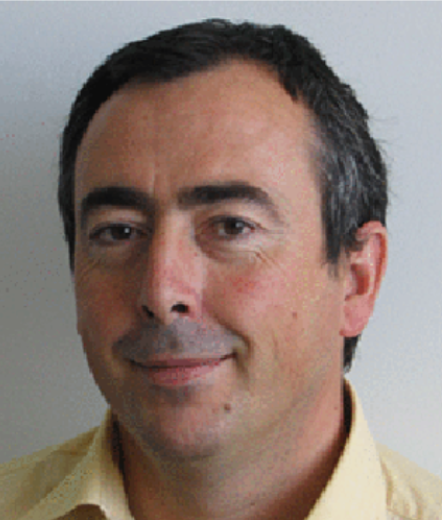}}]{Carlos Sag\"{u}\'{e}s} (Senior Member, IEEE) received the M.Sc. degree in computer science and systems
engineering and the Ph.D. degree in industrial engineering from the University of Zaragoza, Zaragoza, Spain, in 1989 and 1992, respectively.
In 1994, he joined as an Associate Professor with the Departamento de Informatica e Ingenieria de Sistemas, University of Zaragoza, where he became a Full Professor in 2009 and also the Head Teacher. He was engaged in research on force and infrared sensors for robots. His current research interests include control systems and industry applications, computer vision, visual control, and multivehicle cooperative control.
\end{IEEEbiography}

\end{document}